\documentclass[12pt,draftclsnofoot,onecolumn]{IEEEtran}
\usepackage{bm}
\usepackage{amsmath}
\usepackage{amsthm}
\usepackage{extarrows}
\usepackage{amssymb}
\usepackage{graphicx}
\usepackage{color, xcolor}
\usepackage{enumerate}
\usepackage[hidelinks]{hyperref}
\usepackage{bookmark}
\usepackage{setspace}
\usepackage{multirow}
\usepackage{subfigure}
\usepackage{float}
\usepackage{booktabs}
\usepackage{makecell}
\usepackage{algorithm}
\usepackage{algpseudocode}
\usepackage{cite}
\usepackage{framed} 
\usepackage{scalerel}
\usepackage{tikz}
\usetikzlibrary{matrix}
\usepackage{booktabs}

\graphicspath{{figures/}}

\newcommand{\mr}{\mathrm}
\newcommand{\BE}{\begin{equation}}
\newcommand{\EE}{\end{equation}}
\newcommand{\BS}{\begin{subequations}}
\newcommand{\ES}{\end{subequations}}

\newtheorem{assumption}{Assumption}
\newtheorem{definition}{Definition}

\newtheorem{remark}{Remark}
\newtheorem{lemma}{Lemma}

\DeclareRobustCommand{\rchi}{{\mathpalette\irchi\relax}}
\newcommand{\irchi}[2]{\raisebox{\depth}{$#1\chi$}}

\begin{document}

\title{Memory AMP: Overflow Avoidance, Complexity Reduction, and Comparative Analysis}
%\author{\normalsize\IEEEauthorblockN{Shunqi~Huang, \emph{Graduate Student Member, IEEE}, Lei~Liu, \emph{Senior Member, IEEE}, and Brian~M.~Kurkoski, \emph{Member, IEEE}}}
\author{Shunqi~Huang,~\IEEEmembership{Graduate Student Member,~IEEE}, Lei~Liu,~\IEEEmembership{Senior Member,~IEEE}, and Brian~M.~Kurkoski,~\IEEEmembership{Member,~IEEE}%

\thanks{This paper was presented in part at the 2024 IEEE International Symposium on Information Theory (ISIT)~\cite{huang2024overflow}.}%

\thanks{Shunqi Huang and Brian~M.~Kurkoski are with School of Information Science, Japan Advanced Institute of Science and Technology (JAIST), Nomi 923-1292, Japan (e-mail: \{shunqi.huang, kurkoski\}@jaist.ac.jp).}%

\thanks{Lei Liu is with College of Information Science and Electronic Engineering, Zhejiang University (ZJU), Hangzhou
310027, China (e-mail: lei\_liu@zju.edu.cn).}%
}

\maketitle
%\begin{spacing}{1.25}

\begin{abstract}
Approximate message passing (AMP)-type algorithms are widely used for signal recovery in high-dimensional noisy linear systems. Recently, a framework called memory AMP (MAMP) was introduced, offering a new approach to incorporating memory terms within AMP algorithms. Building on this, a low-complexity gradient descent MAMP (GD-MAMP) was proposed for right-unitarily invariant matrices. In this paper, we first address an overflow problem in GD-MAMP caused by intermediate variables exceeding the floating-point range, which typically occurs when the condition number is large. Second, we propose two low-complexity variants of GD-MAMP: one replaces full-length memory with partial memory, while the other reduces the number of matrix-vector products per iteration by $1/3$ (from three to two). Neither degrades the convergence speed notably. Third, we develop a general gradient-based formulation for designing MAMP algorithms. This formulation recovers warm-started conjugate gradient VAMP (WS-CG-VAMP) as a special case. Furthermore, we show that the computation of the orthogonalization parameters in this formulation can suffer from catastrophic cancellation, which explains the finite-precision instability of WS-CG-VAMP. Finally, we derive an equivalent reformulation, termed WS-CG-VAMP(r), which reduces the number of matrix-vector products by up to $50\%$. Measured by matrix-vector products, GD-MAMP converges faster for small condition numbers, whereas WS-CG-VAMP(r) converges faster for large ones under high-precision arithmetic but may diverge in IEEE double precision due to catastrophic cancellation.
\end{abstract}


\section{Introduction}
In this paper, we focus on the signal recovery problem from noisy linear systems, i.e.,
\begin{align}\label{Eqn:system}
    \bm{y} = \bm{Ax} + \bm{n},
\end{align}
where $\bm{x} \in \mathbb{C}^{N \times 1}$ is a signal vector with independent and identically distributed (IID) entries, $\bm{y} \in \mathbb{C}^{M \times 1}$ is a known vector of observations, $\bm{A} \in \mathbb{C}^{M \times N}$ is a known measurement matrix, and $\bm{n}\!\sim\!\mathcal{CN}(\bm{0},\sigma^2\bm{I}_M)$ is a Gaussian noise vector. For general non-Gaussian priors and general matrices, finding the optimal $\bm{x}$ is NP-hard \cite{verdu1985optimum, micciancio2001hardness}.

\subsection{Background}
Approximate message passing (AMP)\cite{donoho2009message}, derived from belief-propagation (BP) with Gaussian approximations, was proposed to solve the problem in (\ref{Eqn:system}). AMP offers several advantages. First, it has low complexity, requiring only $\mathcal{O}(MN)$ computations per iteration due to its use of a matched filter (MF) estimator. Second, the mean squared error (MSE) performance of AMP can be effectively tracked through state evolution (SE) \cite{donoho2009message, bayati2011dynamics}. Furthermore, if $\bm{A}$ is IID Gaussian and the SE has a unique fixed point, AMP is minimum mean squared error (MMSE) optimal (i.e., Bayes-optimal) in an uncoded system  \cite{reeves2019replica, barbier2020mutual}, and achieves the constrained capacity of (\ref{Eqn:system}) in a coded system \cite{liu2021capacity}.

However, AMP may perform poorly or even diverge when dealing with matrices that have correlated entries \cite{manoel2014sparse, vila2015adaptive, rangan2017inference}. To overcome this limitation, orthogonal/vector AMP (OAMP/VAMP) \cite{ma2017orthogonal, rangan2019vector} was proposed for right-unitarily invariant matrices $\bm{A}$. The SE of OAMP/VAMP was proposed in \cite{ma2017orthogonal}, and later proven correct in \cite{rangan2019vector, takeuchi2020rigorous}. Crucially, when the SE has a unique fixed point, it was proven that OAMP/VAMP is replica Bayes-optimal \cite{ma2017orthogonal, rangan2019vector, takeuchi2020rigorous} and replica capacity-optimal\footnote{It is not rigorous to say that OAMP/VAMP is Bayes-optimal and capacity-optimal since the replica method was only proven to be rigorous for some sub-classes of right-unitarily invariant matrices \cite{barbier2018mutual, li2022random}.} \cite{liu2021capacity, liu2021capacity_oamp}. Unfortunately, OAMP/VAMP includes a linear MMSE (LMMSE) estimator, which results in a high per-iteration complexity $\mathcal{O}(M^2N+M^3)$. While performing SVD on $\bm{A}$ can avoid this complexity, SVD itself incurs a complexity of $\mathcal{O}(\min\{M^2N, MN^2\})$. The orthogonality principle is fundamental to OAMP/VAMP. Recently, \cite{liu2023oamp} demonstrated the generality and significance of the orthogonality principle in generic iterative processes. This work shows that AMP (of the form in \cite{takeuchi2019unified}), expectation propagation (EP) \cite{minka2013expectation, opper2005expectation, cakmak2018expectation}, and OAMP/VAMP can be unified under the orthogonality principle.
 
The high complexity of OAMP/VAMP limits its applicability to large-scale systems. To reduce this complexity, convolutional AMP (CAMP) \cite{takeuchi2021bayes} was proposed for right-unitarily invariant matrices, with a low per-iteration complexity of $\mathcal{O}(MN)$. CAMP was also shown to achieve replica Bayes-optimality when its SE converges to a unique fixed point. However, CAMP may diverge for matrices with moderate to large condition numbers. Subsequently, memory AMP (MAMP) \cite{liu2022memory} was introduced as a general framework for designing memory-based AMP-type algorithms, under which AMP, OAMP/VAMP, and CAMP can all be unified. Based on this framework, a Bayes-optimal MAMP (BO-MAMP) algorithm was developed \cite{liu2022memory}. Since BO-MAMP can be derived by using standard gradient descent (GD), we refer to it as gradient descent MAMP (GD-MAMP) in this paper. Like AMP and CAMP, GD-MAMP has a low per-iteration complexity of $\mathcal{O}(MN)$. GD-MAMP achieves replica Bayes-optimality when its SE has a unique fixed point, and its analytically optimized damping guarantees the convergence of its SE. Currently, GD-MAMP has shown its effectiveness in various wireless communication applications, including channel estimation and signal detection for different multicarrier modulation schemes \cite{ge2023low, bian2023joint, chi2024interleave, qi2025MAMP, liu2026random, wang2026channel, chi2026achievable}.

Another notable work is warm-started conjugate gradient VAMP (WS-CG-VAMP) \cite{skuratovs2022compressed, skuratovs2022warm}, which converges faster than GD-MAMP when the condition number is large. However, under finite numerical precision, the computation of its orthogonalization parameters may suffer from numerical errors, causing WS-CG-VAMP to diverge. Around the same time, a low-complexity AMP algorithm for rotationally invariant matrices, referred to as RI-AMP, was developed in \cite{fan2022approximate, venkataramanan2022estimation}. More recently, the divergence-free extension of RI-AMP, referred to as RIAMP-DF, was proposed in \cite{liu2024unifying, luo2026rotationally}. The numerical results in \cite{luo2026rotationally} indicate that GD-MAMP converges slightly more slowly than RI-AMP and RIAMP-DF. However, the optimized damping of GD-MAMP is disabled in that comparison. Since the damping is an essential component of GD-MAMP, the comparison may not fully reflect the convergence performance of these algorithms.

\subsection{Motivation}
As mentioned above, GD-MAMP \cite{liu2022memory} is a low-complexity AMP-type algorithm for right-unitarily invariant matrices. However, several unresolved problems remain, together with opportunities for further complexity reduction:
\begin{itemize}
    \item
    Some intermediate variables in GD-MAMP exhibit exponential growth with the number of iterations. These variables may exceed the representational range of floating-point numbers when the number of iterations is large. This overflow problem typically occurs when the condition number of $\bm{A}$ is large, resulting in the algorithm crashing before convergence. This problem was not clearly identified in \cite{liu2022memory} and requires a theoretical method to resolve.
    \item
    Matrix-vector products dominate the computational cost of GD-MAMP, with each product requiring $\mathcal{O}(MN)$. The original form of GD-MAMP in \cite{liu2022memory} requires four matrix-vector products, which is double that required in AMP. An interesting question is whether we can reduce the number of matrix-vector products in GD-MAMP while preserving the convergence speed. 
\end{itemize}
Beyond improving GD-MAMP itself, it is natural to seek a general gradient-based formulation for designing MAMP algorithms. Such a formulation would clarify the relationship between WS-CG-VAMP~\cite{skuratovs2022compressed} and the MAMP framework. In addition, GD-MAMP and WS-CG-VAMP are both low-complexity AMP-type algorithms for right-unitarily invariant matrices. A fair comparison of their convergence efficiency is important, and should be based on the number of matrix-vector products. Furthermore, the source of the finite-precision instability in WS-CG-VAMP remains to be clarified.

\subsection{Contributions}
In this paper, we first solve the overflow problem of GD-MAMP. Then, we develop two variants of GD-MAMP to further reduce its complexity:
\begin{itemize}
    \item
    We develop an overflow-avoiding GD-MAMP (OA-GD-MAMP) using the eigenvalues of $\bm{A}\bm{A}^{\rm H}$, which maintains equivalence to GD-MAMP while solving the overflow problem. When the eigenvalues of $\bm{A}\bm{A}^{\rm H}$ are unknown, we introduce an approximation method to develop an alternative form of OA-GD-MAMP. This alternative form achieves performance comparable to the OA-GD-MAMP with known eigenvalues.
    \item
    We reduce the number of matrix-vector products per iteration in the original GD-MAMP from four to three by eliminating redundant computations. Subsequently, we develop two variants of GD-MAMP to further reduce the complexity. The first variant, partial GD-MAMP, reduces the space complexity by eliminating the requirement to store full-length memory estimates. More significantly, the second variant, complexity-reduced GD-MAMP (CR-GD-MAMP), only requires two matrix-vector products per iteration, reducing the complexity by $1/3$.
\end{itemize}
Beyond these improvements to GD-MAMP, we develop a general
gradient-based formulation for designing MAMP algorithms and
study WS-CG-VAMP within this formulation:
\begin{itemize}
    \item We develop a general gradient-based formulation for
    designing MAMP algorithms and show that WS-CG-VAMP can be
    recovered by specializing it to the CG method. We establish Krylov-type polynomial representations for the inner iterates and provide recursive constructions of the corresponding coefficients. Based on these results, we show that the computation of the orthogonalization parameters may suffer from catastrophic cancellation. We derive a relative roundoff-error bound, thereby explaining the finite-precision instability of WS-CG-VAMP.
    \item We derive an equivalent reformulation of WS-CG-VAMP, termed WS-CG-VAMP(r), which eliminates redundant computations. Relative to a direct implementation of WS-CG-VAMP, WS-CG-VAMP(r) reduces the number of matrix--vector products per iteration from $2I+4$ to $2I+1$, where $I$ denotes the number of inner CG iterations. Using the number of matrix-vector products as the measure, our numerical results show that GD-MAMP converges faster for small condition numbers, whereas WS-CG-VAMP(r) converges faster for large condition numbers under high-precision arithmetic but may diverge in IEEE double precision due to catastrophic cancellation.
\end{itemize}
A preliminary version of part of this work was published in
\cite{huang2024overflow}.

\subsection{Notation}
Boldface lowercase and boldface uppercase symbols denote vectors and matrices respectively. A matrix is called column-wise IID Gaussian and row-wise joint Gaussian if each column is IID Gaussian and each row is joint Gaussian. For $N \in \mathbb{N}^*$, $[N]$ denotes the set $\{1, \cdots\!, N\}$. The notation $\bm{V} \equiv [v_{i,j}]_{N \times N}$ represents an $N \times N$ matrix $\bm{V}$ with $(i,j)$-th entry $v_{i,j}$. For $c \in \mathbb{C}$, $c^{*}$ denotes the conjugate of $c$. For vectors or matrices, $(\cdot)^{\rm T}$ and $(\cdot)^{\rm H}$ denote the transpose and Hermitian transpose, respectively. The notation $\overset{\rm a.s.}{=}$ denotes almost sure convergence. Unless specified otherwise, as an algorithmic step, $a \overset{\rm a.s.}{=} \lim_{N \to \infty} b$ can be simplified to $a = b$.

\section{Preliminaries}
In this section, we begin with introducing the problem model and outlining the underlying assumptions. Then, we briefly review the MAMP framework and the GD-MAMP algorithm.

\subsection{Problem Formulation and Assumptions}\label{Sec:pre1}
We rewrite the linear system in (\ref{Eqn:system}) with two constraints as   
\BS\label{Eqn:system2}\begin{align}
&\Gamma: \quad  \bm{y}=\bm{Ax}+\bm{n},\\
&\Phi: \quad x_i\sim P_x, \;\;\forall i.
\end{align}\ES 
The MMSE of the system in (\ref{Eqn:system2}) is defined as
\begin{align}\label{Eqn:mmse}
    {\rm mmse}\{\bm{x}|\bm{y}, \bm{A}, \Gamma, \Phi\} \equiv \tfrac{1}{N}{\mr{E}}\{\|\hat{\bm{x}}_{\rm post}-{\bm{x}}\|^2\},
\end{align}
where $\hat{\bm{x}}_{\rm post}={\mr{E}}\{\bm{x}|\bm{y}, \bm{A}, {\Gamma}, {\Phi} \}$ is the posterior mean of $\bm{x}$. An iterative algorithm is said to be Bayes-optimal if its MSE converges to the MMSE in (\ref{Eqn:mmse}) \cite{kay1993fundamentals}. 

The assumptions of this paper are as follows:
\begin{assumption} \label{Assum:1}
    The entries of $\bm{x}$ are IID with zero mean and normalized variance, i.e., ${\rm E}\{\bm{x}\}=\bm{0}$ and $\tfrac{1}{N}{\rm E}\{\|\bm{x}\|^2\}=1$. Furthermore, the $(2+i)$-th moments of $\bm{x}$ are finite for some $i \in \mathbb{N}^*$.
\end{assumption}
\begin{assumption} \label{Assum:2}
    The large-system limit is assumed, where $M, N \to \infty$ with a fixed $\delta=M/N \in(0,\infty)$. In addition, $\bm{A}$ is a right-unitarily invariant matrix, meaning that in its singular value decomposition $\bm{A} = \bm{U\Sigma}\bm{V}^{\rm H}$, $\bm{V}$ is Haar distributed and independent of $\bm{U}$ and $ \bm{\Sigma}$.
\end{assumption}
Note that we consider the conventional proportional-growth regime, where $M,N\to\infty$ with $M/N=\delta$, and assume an $N$-independent i.i.d.\ signal prior. This differs from the sublinear-sparsity regime recently studied in
\cite{takeuchi2025generalized,takeuchi2025generalized_2,takeuchi2026direct}, which is outside the scope of this paper.

Unless specified otherwise, we assume that the extremal eigenvalues $\lambda_{\rm max}$ and $\lambda_{\rm min}$ of $\bm{A}\bm{A}^{\rm H}$ are known. In addition, its moments, $\lambda_t \equiv \tfrac{1}{N}{\rm tr}\{(\bm{A}\bm{A}^{\rm H})^t\}$ for $t \in [T']$ with a sufficiently large $T'$, are also known. When these quantities are unavailable, the approximation method in \cite{liu2022memory} can be employed:
\begin{itemize}
\item{$\{\lambda_t\}$ can be approximated by
\begin{align}\label{Eqn:lam}
    \lambda_{t}  \overset{\mr{a.s.}}{=} \lim_{N\to \infty} \|\bm{s}_{t}\|^2,
\end{align}
where $\bm{s}_t$ is given by a recursion: starting with $t=1$ and $\bm{s}_0\sim \mathcal{CN}(\bm{0},\tfrac{1}{N}\bm{I}_{N})$,
\begin{align}\label{Eqn:app_s}
    \bm{s}_t = 
    \begin{cases}
        \bm{A}\bm{s}_{t-1}, & \ {\rm if\;} t {\rm \;is\; odd} \\[1mm]
        \bm{A}^{\rm H} \bm{s}_{t-1}, & \ {\rm if\;} t {\rm \;is\; even} 
    \end{cases}.
\end{align}
}
\item{$\lambda_{\rm max}$ and $\lambda_{\rm min}$ can be replaced respectively by $\lambda_{\rm max}^{\rm up}$ and $\lambda_{\rm min}^{\rm low}$, where 
\begin{align}\label{Eqn:max_min}
    \lambda_{\rm max}^{\rm up} = (N \lambda_\tau)^{1/\tau},\ \lambda_{\rm min}^{\rm low} = 0.
\end{align}
The upper bound $\lambda_{\rm max}^{\rm up}$ becomes tighter as $\tau$ increases.
}
\end{itemize}

\subsection{Memory AMP}
Memory AMP (MAMP) \cite{liu2022memory} is a general algorithmic framework for designing AMP-type algorithms with memory-based estimators. In many existing works, the term ``MAMP'' specifically refers to the GD-MAMP algorithm \cite{liu2022memory}, which is reviewed in Section \ref{Sec:GD-MAMP}. To avoid confusion, we explicitly distinguish between the general MAMP framework and specific algorithms throughout this paper.
\begin{definition}[Memory AMP \cite{liu2022memory}]\label{def:MAMP}
An iterative process is said to be a memory AMP (MAMP) algorithm if it consists of a memory linear estimator (MLE) and a memory nonlinear estimator (MNLE): starting with iteration index $t=1$, $\bm{x}_1 = {\rm E}\{\bm{x}\} = \bm{0}$,
\BS \label{Eqn:gen_MAMP}
\begin{alignat}{2}
    {\rm MLE:} && \quad \bm{r}_t &= \gamma_t(\bm{X}_t) = \bm{Q}_t\bm{y} + \textstyle\sum_{i=1}^t{\bm{P}}_{t,i} \bm{x}_i, \label{Eqn:gen_MLE} \\
    {\rm MNLE:} && \quad \bm{x}_{t+1} &= \phi_t(\bm{R}_t), \label{Eqn:gen_MNLE}
\end{alignat}
\ES
where $\bm{X}_t = [\bm{x}_1, \cdots\!, \bm{x}_t]$, $\bm{R}_t = [\bm{r}_1, \cdots\!, \bm{r}_t]$, and 
\begin{itemize}
    \item Let $\bm{f}_t = \bm{x}_t - \bm{x}$ and $\bm{g}_t = \bm{r}_t - \bm{x}$ denote the error vectors of the estimates $\bm{x}_t$ and $\bm{r}_t$. The following orthogonality constraints hold:
    for $k \in [t]$,
    \BS\label{Eqn:orth}\begin{align}
        &\lim_{N\to\infty} \tfrac{1}{N} \bm{g}_t^{\rm H}\bm{f}_{k} \overset{\rm a.s.}{=} 0, \label{Eqn:orth_1} \\
        &\lim_{N\to\infty} \tfrac{1}{N} \bm{g}_t^{\rm H}\bm{x} \overset{\rm a.s.}{=} 0, \label{Eqn:orth_2} \\
        &\lim_{N\to\infty} \tfrac{1}{N} \bm{f}_{t+1}^{\rm H}\bm{g}_{k} \overset{\rm a.s.}{=} 0. \label{Eqn:orth_3}
    \end{align}\ES
    An MLE is said to be orthogonal if \eqref{Eqn:orth_1} and \eqref{Eqn:orth_2} hold, and an MNLE is said to be orthogonal if \eqref{Eqn:orth_3} holds. In other words, a MAMP algorithm consists of an orthogonal MLE and an orthogonal MNLE.
    \item $\bm{Q}_t\bm{A}$ and $\bm{P}_{t,i}$ are polynomials in $\bm{A}^{\rm H} \bm{A}$, i.e., $\bm{Q}_t\bm{A} = \sum_{k\geq1}\alpha_{t,k}(\bm{A}^{\rm H}\bm{A})^k$ and $\bm{P}_{t,i} = \sum_{k\geq0}\beta_{t,k}(\bm{A}^{\rm H}\bm{A})^k$, where $\{\alpha_{t,k}\}, \{\beta_{t,k}\}$ are real-valued coefficients. To enforce the orthogonality constraints in \eqref{Eqn:orth_1} and \eqref{Eqn:orth_2}, the following conditions are required:
    \BS\begin{align}
        & \tfrac{1}{N} {\rm tr} \{\bm{Q}_t\bm{A} \} =1, \label{Eqn:Q}  \\ 
        & {\rm tr} \{ \bm{P}_{t, i} \} =0,\  \forall i \leq t, \label{Eqn:P}
    \end{align}\ES
    Without loss of generality, we assume that the spectral norms of $\bm{Q}_t$ and $\bm{P}_{t,i}$ are finite. Hence, $\gamma_t(\cdot)$ is Lipschitz-continuous \cite{liu2022memory}. 
\end{itemize}
\end{definition}

For $ i,j \in [t]$, we define the normalized cross-error covariances (or the normalized error variance when $i=j$) as
\BS\begin{align}
    v^{{\phi}}_{i,j} &\equiv \tfrac{1}{N}{\mr E} \{\bm{f}_{i}^{\mr H} \bm{f}_{j}\}, \\
    v^{\gamma}_{i,j} &\equiv \tfrac{1}{N}{\mr E} \{\bm{g}_{i}^{\mr H} \bm{g}_{j}\},
\end{align}\ES
where $\bm{f}_i = \bm{x}_i - \bm{x}$ and $\bm{g}_i = \bm{r}_i - \bm{x}$ are the error vectors. These terms are collected into the error covariance matrices $\bm{V}^{\phi}_{t} \equiv [v^{{\phi}}_{i,j}]_{t \times t}$ and $\bm{V}^{\gamma}_{t} \equiv [v^{{\gamma}}_{i,j}]_{t \times t}$. For simplicity, in the remainder of this paper, we refer to $v_{i, j}^{\phi}$ as the covariance between $\bm{x}_i$ and $\bm{x}_j$, and $v_{t, t}^{\phi}$ as the variance of $\bm{x}_t$. The following lemma shows the asymptotic Gaussianity of the error vectors in a MAMP algorithm.

\begin{lemma}[Asymptotic Gaussianity \cite{takeuchi2021bayes, takeuchi2019unified}]\label{Lemma:AIIDG}
    Suppose that Assumptions \ref{Assum:1}-\ref{Assum:2} hold. For a MAMP in \eqref{Eqn:gen_MAMP}, if the MLE $\gamma_t(\cdot)$ is Lipschitz-continuous, and the MNLE $\phi_t(\cdot)$ is separable and Lipschitz-continuous, then for $k \in [t]$,
    \BS\begin{flalign}
        \lim_{N\to\infty} v_{t,k}^{\gamma} &\overset{\rm a.s.}{=} \tfrac{1}{N}\mr{E}\big\{[\gamma_t(\bm{X}+\bm{N}_t^\phi)-\bm{x}]^{\rm H} [\gamma_{k}(\bm{X}+\bm{N}_k^\phi)-\bm{x}]\big\}, \\
        \lim_{N\to\infty} v_{t+1,k+1}^{\phi} &\overset{\rm a.s.}{=} \tfrac{1}{N}\mr{E}\big\{[\phi_t(\bm{X}+\bm{N}_t^\gamma)-\bm{x}]^{\rm H} [\phi_{k}(\bm{X}+\bm{N}_k^\gamma)-\bm{x}]\big\},
    \end{flalign}\ES
    where $\bm{X} = [\bm{x}, \cdots\!,\bm{x}]$ has $t$ columns, 
    $\bm{N}_t^\phi = [\bm{n}_1^\phi, \cdots\!, \bm{n}_t^\phi]$ and 
    $\bm{N}_t^\gamma =[\bm{n}_1^\gamma, \cdots\!, \bm{n}_t^\gamma]$ are independent of ${\bm{n}}$ and ${\bm{x}}$, with IID Gaussian columns and jointly Gaussian rows satisfying
    \BS\begin{align}
        &\bm{n}_t^\phi \sim \mathcal{CN}(\bm{0}, v_{t,t}^\phi\bm{I}), \  {\rm E}\{\bm{n}_t^\phi(\bm{n}_k^\phi)^{\rm H}\} = v_{k,t}^\phi\bm{I}, \\
        &\bm{n}_t^\gamma \sim \mathcal{CN}(\bm{0}, v_{t,t}^{\gamma}\bm{I}), \  {\rm E}\{\bm{n}_t^\gamma(\bm{n}_k^\gamma)^{\rm H}\} = v_{k,t}^{\gamma}\bm{I}.
    \end{align}\ES
\end{lemma}
Note that Lemma \ref{Lemma:AIIDG} does not state that the error matrices $\bm{F}_t = [\bm{f}_1, \cdots\!, \bm{f}_t]$ and $\bm{G}_t = [\bm{g}_1, \cdots\!, \bm{g}_t]$ are themselves Gaussian, as such a claim would be too strong. Instead, it establishes an asymptotic equivalence: for the specific purpose of computing the output variances or covariances ($v^{\gamma}$ and $v^{\phi}$), the true error matrices $\bm{F}_t$ and $\bm{G}_t$ behave as if they were drawn from the specified Gaussian distributions. In other words,
$\bm{F}_t$ and $\bm{G}_t$ can be replaced by the column-wise IID and row-wise jointly Gaussian matrices $\bm{N}_t^\phi$ and $\bm{N}_t^\gamma$, respectively. This property enables a deterministic and simplified state evolution analysis to track the performance of a MAMP algorithm.

In practice, we often start with general MLEs or MNLEs that may not satisfy the orthogonality constraints in \eqref{Eqn:orth}. The following lemma presents an orthogonalization procedure to construct a MAMP from such general estimators.
\begin{lemma}[MAMP Construction \cite{liu2022memory}]\label{Lemma:MAMP}
    Given a general MLE:
    \begin{align}
        \hat{\gamma}_t(\bm{X}_t) = \hat{\bm{Q}}_t\bm{y} + \textstyle\sum_{i=1}^t\hat{\bm{P}}_{t,i} \bm{x}_i,
    \end{align}
    where $\hat{\bm{Q}}_t\bm{A}$ and $\hat{\bm{P}}_{t,i}$ are polynomials in $\bm{A}^{\rm H} \bm{A}$, and an arbitrary differentiable, separable, and Lipschitz-continuous MNLE $\hat{\phi}_t(\cdot)$, then, we can construct an orthogonal MLE $\gamma_t(\cdot)$ and the orthogonal MNLE $\phi_t(\cdot)$ as follows:
    \BS
    \begin{align}
        \gamma_t(\bm{X}_t) &= \tfrac{1}{\varepsilon_t^{\gamma}} \big[\hat{\gamma}_t(\bm{X}_t)-\bm{X}_t\bm{p}^{\gamma}_t\big], \\
        \phi_t(\bm{R}_t) &= \tfrac{1}{\varepsilon_t^{\phi}} \big[\hat{\phi}_t(\bm{R}_t)-\bm{R}_t \bm{p}^{\phi}_t\big],
    \end{align}
    where 
    \begin{align}
        \varepsilon_t^{\gamma} &= \tfrac{1}{N}{\rm tr}\{\hat{\bm{Q}}_t\bm{A}\}, \\
        \bm{p}^{\gamma}_t &= \big[\tfrac{1}{N}{\rm tr}\{\hat{\bm{P}}_{t,1}\}, \cdots\!, \tfrac{1}{N}{\rm tr}\{\hat{\bm{P}}_{t,t}\}\big]^{\rm T}, \\
        \bm{p}^{\phi}_t &= \big[{\rm E}\{\tfrac{\partial \hat{\phi}_t }{\partial r_1}\}, \cdots\!, {\rm E}\{\tfrac{\partial \hat{\phi}_t}{\partial r_t}\} \big]^{\rm T},
    \end{align}
    \ES
    and $\varepsilon_t^{\phi}$ is a scaling constant, typically determined by minimizing the MSE of $\phi_t(\cdot)$.
\end{lemma}

\subsection{Gradient Descent MAMP}\label{Sec:GD-MAMP}
A low-complexity Bayes-optimal MAMP algorithm was developed in \cite{liu2022memory}. It is referred to as gradient descent MAMP (GD-MAMP) in this paper, since it can be derived using a scaled gradient descent method.

\textbf{\emph{GD-MAMP Algorithm}} 
\cite{liu2022memory}: Let ${\lambda}^\dag = ( \lambda_{\max} + \lambda_{\min}) / 2$ and $\bm{B} = \lambda^\dag\bm{I} - \bm{A}\bm{A}^{\mr H}$. Starting with $t=1$, $\bm{u}_{0} = \bm{0}_M$, $\bm{x}_1={\rm E}\{\bm{x}\}=\bm{0}_N$,
\BS\label{Eqn:GD-MAMP}
\begin{alignat}{2}
    {\rm MLE:} && \bm{u}_{t} &= \theta_t \bm{B} \bm{u}_{t-1} + \xi_t(\bm{y} - \bm{A}\bm{x}_t), \label{Eqn:MLE1}\\
    && \quad \quad\bm{r}_t &=\gamma_t (\bm{X}_t) = \tfrac{1}{{\varepsilon}^\gamma_t}\big(\bm{A}^{\mr H}\bm{u}_{t} + \textstyle\sum_{i=1}^t  p_{t, i}\bm{x}_i \big), \label{Eqn:MLE2}\\
    {\rm NLE:} && \quad  \bm{x}_{t + 1} &= \bar\phi_t(\bm{r}_t) = \big[\bm{x}_1, \cdots\!, \bm{x}_t, \phi_t ( \bm{r}_t)\big] \cdot \bm{\zeta}_{t+1}. \label{Eqn:NLE}
\end{alignat}\ES
The details of the MLE are as follows:
\begin{enumerate}
    \item
    The relaxation parameter $\theta_t$ is optimized by
    \begin{align}\label{Eqn:theta}
        \theta_t = (\lambda^\dag + \rho_t)^{-1}.
    \end{align}
    where $\rho_t = \sigma^2 / v_{t,t}^{\bar{\phi}}$ with $v_{1,1}^{\bar{\phi}} = (\tfrac{1}{N} \|\bm{y}-\bm{A}\bm{x}_1\|^2 - \delta \sigma^2)/w_0$ and $v_{t,t}^{\bar{\phi}}$ for $t \geq 2$ given in (\ref{Eqn:v_phi_bar}).
    \item
    For $i \geq 0$, $j \geq 0$, let\footnote{In practice, we do not compute $b_i$ as in \eqref{Eqn:b}. Instead, if the eigenvalues of $\bm{AA}^{\rm H}$ are known, we use $b_i = \tfrac{1}{N}\sum_{k=1}^M \hat{\lambda}_k^i$, where $\{\hat{\lambda}_k\}$ are the eigenvalues of $\bm{B} = \lambda^\dag\bm{I} - \bm{A}\bm{A}^{\mr H}$; otherwise, we approximate $b_i$ in a manner analogous to \eqref{Eqn:app_s}.}
    \begin{align}
        b_i &\equiv \tfrac{1}{N}{\rm tr}\{\bm{B}^i\} = \textstyle\sum_{k=0}^{i} \binom{i}{k} (-1)^k(\lambda^\dag)^{i-k}\lambda_k, \label{Eqn:b} \\
        w_i &\equiv  \tfrac{1}{N}{\rm tr}\{\bm{A}^{\rm H}\bm{B}^{i}\bm{A}\} = \lambda^\dag b_i - b_{i+1}, \label{Eqn:w}\\
        \bar{w}_{i,j} &= \lambda^\dag  w_{i+j}-w_{i+j+1}-w_{i}w_{j}, \\
        \vartheta_{t, i} &= \xi_i \textstyle\prod_{j=i+1}^t\theta_j,\ \;i \in [t],
    \end{align}
    % Let
    % \begin{align}\label{Eqn:vartheta}
    %     \vartheta_{t, i} =
    %     \begin{cases}
    %         \xi_t, & \ i = t \\[1mm]
    %         \xi_i \textstyle\prod_{\tau=i+1}^t\theta_\tau, & \ 1 \leq i < t
    %     \end{cases},
    % \end{align}
    where the weight parameter $\xi_t$ is optimized as: $\xi_1 = 1$ and
    \begin{align}\label{Eqn:xi}
        \xi_t = \frac{c_{t,2}c_{t,0}+c_{t,3}}{c_{t,1}c_{t,0}+c_{t,2}},\ \;t > 1,
        % \begin{cases}
        %     \dfrac{c_{t,2}c_{t,0}+c_{t,3}}{c_{t,1}c_{t,0}+c_{t,2}}, & \ t > 1 \\[1mm]
        %     1, & \ t = 1
        % \end{cases},
    \end{align}
    with
    \BS\label{Eqn:ct0_4}\begin{align}
        c_{t,0} &= \textstyle\sum_{i=1}^{t-1} \vartheta_{t, i} w_{t-i} / w_0,\\
        c_{t,1} &= \sigma^2  w_0+ v_{{t},{t}}^{\bar{\phi}}\bar{ w}_{0,0},\\
        c_{t,2} &= - \textstyle{\sum}_{i=1}^{t-1} \vartheta_{t, i} \big(\sigma^2 w_{t-i}+ \Re(v_{{t},{i}}^{\bar{\phi}})\bar{ w}_{0,t-i}\big), \\
        c_{t,3} &= \textstyle{\sum}_{i=1}^{t-1}\textstyle{\sum}_{j=1}^{t-1} \vartheta_{t, i}\vartheta_{t,j}\big(\sigma^2 w_{2t-i-j} + v^{\bar{\phi}}_{i,j}\bar{w}_{t-i,t-j}\big).
    \end{align}\ES
    The orthogonalization parameters $\{p_{t,i}\}$ and  normalization factor ${\varepsilon}^\gamma_t$ are given by
   \begin{align}
        p_{t, i} &= \vartheta_{t, i} w_{t-i}, \label{Eqn:p_ti} \\ 
        {\varepsilon}^\gamma_t &= \textstyle\sum_{i=1}^t p_{t, i}.
    \end{align}
    \item
    For $k \in [t]$, $v_{t, k}^{\gamma}$ is given by:
    \begin{align}\label{Eqn:v_gam}
        v^{\gamma}_{t, k} = \tfrac{1}{\varepsilon^\gamma_{t}\varepsilon^\gamma_{k}} \textstyle{\sum}_{i=1}^{t}\textstyle{\sum}_{j=1}^{k}  \vartheta_{t, i}\vartheta_{k,j}\big[\sigma^2  w_{t+k-i-j} + v^{\bar{\phi}}_{i,j}\bar{ w}_{t-i,k-j}\big].
    \end{align}
\end{enumerate}
The details of the NLE are as follows:
\begin{enumerate}
    \item{
    $\phi_t(\cdot)$ is a separable and Lipschitz-continuous function, which is the same as the NLE in OAMP/VAMP \cite{ma2017orthogonal}, \cite{rangan2019vector}. 
    }
    \item{
    For $t \geq 2$, 
    \begin{align}\label{Eqn:v_phi}
    v^{\phi}_{t,k} \overset{\rm a.s.}{=}
    \begin{cases}
        \lim\limits_{N\to\infty}(\tfrac{1}{N} \|\bar{\bm{z}}_t\|^2 - \delta \sigma^2)/w_0, & k = t \\
        \lim\limits_{N\to\infty}(\tfrac{1}{N} \bar{\bm{z}}_t^{\mr H} \bm{z}_{k} - \delta \sigma^2)/w_0 , & 1 \leq k < t 
    \end{cases}.
    \end{align}
    where $\bar{\bm{z}}_t = \bm{y}-\bm{A}\,\phi_{t-1}(\bm{r}_{t-1})$ and $\bm{z}_{k} = \bm{y}-\bm{A}\bm{x}_{k}$.
    }
    \item{
    Let $\bm{V}_{t+1}^{\phi}$ be the error covariance matrix for $\bm{x}_1, \cdots\!, \bm{x}_{t}, \phi_t(\bm{r}_t)$, which is assumed to be invertible\footnote{If $\bm{V}_{t+1}^{\phi}$ is singular or ill-conditioned, we employ the back-off damping in \cite{liu2022memory}, or add a small amount of diagonal loading.}. The optimal damping vector is given by 
    \begin{align}\label{Eqn:xi_2}
        \bm{\zeta}_{t+1} = 
        \frac{(\bm{V}_{t+1}^{\phi})^{-1} \bm{1}}{\bm{1}^{\rm T} (\bm{V}_{t+1}^{\phi})^{-1}\bm{1}}.
    \end{align}
    Then, for $k \in [t+1]$, 
    \begin{align}\label{Eqn:v_phi_bar}
        v^{\bar{\phi}}_{t+1,k} = v^{\bar{\phi}}_{k,t+1} = 1 \;/\; \bm{1}^{\rm T} (\bm{V}_{t+1}^{\phi})^{-1}\bm{1}.
    \end{align}
    It means that the error covariance matrix $\bm{V}_{t+1}^{\bar{\phi}}$ for the damped estimates $\bm{x}_1, \cdots\!, \bm{x}_{t+1}$ is an L-banded matrix \cite{huang2023algebra} (referred to as an L-matrix in \cite{bouthat2021matrices, vstampach2022hilbert, vstampach2022asymptotic}). This special structure ensures the convergence of the SE of GD-MAMP, and allows $(\bm{V}_{t+1}^{\phi})^{-1}\bm{1}$ to be computed in  $\mathcal{O}(t)$ time \cite{takeuchi2022convergence,liu2022sufficient, liu2021sufficient}. In practice, we typically set the maximum damping length as $L=2$ or $L=3$ to improve the robustness. That is, $\bm{x}_{t+1} = \big[\bm{x}_{t-\ell+2}, \cdots\!, \bm{x}_t, \phi_t ( \bm{r}_t)\big] \cdot \bm{\zeta}_{t+1}$, where $\ell = \min\{t+1, L\}$ and $\bm{\zeta}_{t+1}$ is obtained by replacing $\bm{V}_{t+1}^{\phi}$ with its bottom-right $\ell \times \ell$ submatrix in \eqref{Eqn:xi_2} and \eqref{Eqn:v_phi_bar}.
    }
\end{enumerate}

\subsection{Simplification of GD-MAMP}\label{SSec:SGD}
The time complexity of GD-MAMP is dominated by the number of matrix-vector products. The original implementation of GD-MAMP in \eqref{Eqn:GD-MAMP} requires four matrix-vector products per iteration, outlined as follows\footnote{Since $\bm{z}_t = \bm{y} - \bm{A}\bm{x}_t = [\bm{z}_1, \cdots\!, \bm{z}_{t-1}, \bar{\bm{z}}_t] \cdot \bm{\zeta}_{t}$, computing it in \eqref{Eqn:MLE1} does not require a matrix-vector product.}: 
\begin{itemize}
    \item Computing $\bar{\bm{z}}_t = \bm{y}-\bm{A}\phi_{t-1}(\bm{r}_{t-1})$ in \eqref{Eqn:v_phi} requires one.
    \item Computing $\bm{B}\bm{u}_{t-1} = (\lambda^\dag\bm{I} - \bm{A}\bm{A}^{\mr H})\bm{u}_{t-1}$ in (\ref{Eqn:MLE1}) requires two.
    \item Computing $\bm{A}^{\mr H}\bm{u}_{t}$ in (\ref{Eqn:MLE2}) requires one.
\end{itemize}
Notice that the computation of $\bm{A}^{\rm H}\bm{u}_{t-1}$ is redundant, since it was already computed in the iteration $t-1$. Therefore, we can rewrite (\ref{Eqn:GD-MAMP}) as: starting with $t = 1$, $\bm{x}_1 = \hat{\bm{r}}_0 = \bm{0}_{N}$ and $\bm{u}_{0} = \bm{0}_{M}$,
\BS
\label{Eqn:SGD_MAMP}
\begin{alignat}{2}
    {\rm MLE:} && \bm{u}_{t} &= \theta_t\lambda^\dag\bm{u}_{t-1} + \xi_t\bm{y} - \bm{A}(\theta_t\hat{\bm{r}}_{t-1} + \xi_t\bm{x}_t), \label{Eqn:SGD_MAMP_a}\\
    && \hat{\bm{r}}_{t} &= \bm{A}^{\mr H}\bm{u}_{t} \label{Eqn:SGD_MAMP_b}\\
    && \bm{r}_t &= \gamma_t(\bm{X}_t) = \tfrac{1}{{\varepsilon}^\gamma_t}\big(\hat{\bm{r}}_{t} + \textstyle\sum_{i=1}^t  p_{t, i}\bm{x}_i \big), \label{Eqn:SGD_MAMP_c}\\
    {\rm NLE:} && \quad \bm{x}_{t+1} &= \bar\phi_t(\bm{r}_t) = \big[\bm{x}_1, \cdots\!, \bm{x}_t, \phi_t ( \bm{r}_t)\big] \cdot \scaleto{\bm{\zeta}}{8pt}_{t+1}. 
\end{alignat}
\ES
As a result, the number of matrix-vector products required in each iteration of GD-MAMP can be decreased from \emph{four} to \emph{three}, leading to an almost $1/4$ reduction in complexity. This reduction is relatively trivial. In Section~\ref{Sec:CR_GD}, we propose a variant of GD-MAMP, further reducing the number of matrix-vector products per iteration from \emph{three} to \emph{two} without degrading the convergence speed notably. 

\section{Overflow-Avoiding GD-MAMP}\label{Sec:OA-GD}
In this section, we first analyze the risk of overflow in GD-MAMP. To address this problem, we reformulate the original GD-MAMP to obtain an overflow-avoiding GD-MAMP (OA-GD-MAMP), which uses the eigenvalues of $\bm{A}\bm{A}^{\rm H}$. When these eigenvalues are unknown, we propose an approximation method to ensure the effectiveness of OA-GD-MAMP.

\subsection{Overflow Issue in GD-MAMP}
Numerical overflow is a fundamental issue in computer arithmetic, which occurs when a calculated value exceeds the maximum representable limit of its data type. For instance, the double-precision floating-point format can store a maximum value of approximately $1.8 \times 10^{308}$. This leads to a computational failure.

In GD-MAMP, the risk of overflow originates from the calculation of $w_t = \tfrac{1}{N}{\rm tr}\{\bm{A}^{\rm H}\bm{B}^t\bm{A}\}$, which is required for up to $t = 2T$ with $T$ being the maximum number of iterations. When the spectral radius of $\bm{B}$ is greater than $1$, the value of $w_t$ increases exponentially with $t$. For a highly ill-conditioned matrix $\bm{A}$, GD-MAMP requires a large $T$ for convergence. In such cases, $w_t$ may exceed the representable limit, leading to a collapse of the algorithm. 

\subsection{OA-GD-MAMP with Known Eigenvalues of \texorpdfstring{$\bm{A}\bm{A}^{\rm H}$}{TEXT}}
In GD-MAMP, we observe that $w$ is always paired with the corresponding $\vartheta$, for example, $\vartheta_{t,i}w_{t-i}$ in \eqref{Eqn:p_ti} and $\vartheta_{t,i}\vartheta_{t,j}w_{2t-i-j}$ in \eqref{Eqn:v_gam}. Although $w$ can grow exponentially, the product term $\vartheta w$ remains bounded by a moderate constant. For instance, as $w_{t-i}$ becomes large, $\vartheta_{t,i}$ approaches zero, ensuring that $|\vartheta_{t,i} w_{t-i}|$ remains bounded. This observation motivates us to find a method to calculate $\vartheta w$ directly, without explicitly calculating $w$ and $\vartheta$, thereby avoiding potential overflow and underflow respectively.

We define some element-wise operations as follows: for $\bm{a}=[a_1, \cdots\!, a_k]^{\rm T} \in \mathbb{R}^k$, $\bm{b}=[b_1, \cdots\!, b_k]^{\rm T}\in \mathbb{R}^k$, and $m \in \mathbb{R}$,
\BS\begin{align}
    & |\bm{a}| \equiv \big[|a_1|, \cdots\!, |a_k| \big]^{\rm T},\ {\rm sgn}(\bm{a}) \equiv \big[{\rm sgn}(a_1),\cdots\!, {\rm sgn}(a_k)\big]^{\rm T}, \\
    & m^{\circ \bm{a}} \equiv [m^{a_1}, \cdots\!, m^{a_k}]^{\rm T},\ \log^{\circ}(\bm{a}) \equiv \big[\log(a_1), \cdots\!, \log(a_k)\big]^{\rm T}, \\
    & \bm{a} \circ \bm{b} \equiv [a_1 b_1, \cdots\!, a_k b_k]^{\rm T},\ \bm{a}^{\circ m} \equiv [a_1^m, \cdots\!, a_k^m]^{\rm T}.
\end{align}\ES
%$|\bm{a}| \equiv \big[|a_1|, \cdots\!, |a_k| \big]^{\rm T}$, ${\rm sgn}(\bm{a}) = \big[{\rm sgn}(a_1),\cdots\!, {\rm sgn}(a_k)\big]^{\rm T}$, $\log^{\circ}(\bm{a}) = \big[\log(a_1), \cdots\!, \log(a_k)\big]^{\rm T}$, $\mathrm{e}^{\circ \bm{a}} \equiv [\mathrm{e}^{a_1}, \cdots\!, \mathrm{e}^{a_k}]^{\rm T}$, $\bm{a} \circ \bm{b} \equiv [a_1 b_1, \cdots\!, a_k b_k]^{\rm T}$, and $\bm{a}^{\circ m} = [a_1^m, \cdots\!, a_k^m]^{\rm T}$.
Consider that the eigenvalues of $\bm{A}\bm{A}^{\rm H}$ are known. Then, the vector of eigenvalues of $\bm{B} = \lambda^\dag\bm{I} - \bm{A}\bm{A}^{\mr H}$, denoted as $\bm{\lambda}_{\bm{B}}$, is also known. For any non-zero $\alpha$, if ${\rm sgn}(\alpha)$ and $\log|\alpha|$ are known, the following lemma provides a method to calculate $\alpha w_k$ without calculating $\alpha$ and $w_k$ individually. 
\begin{lemma}\label{Lemma:oa_1}
    For any $\alpha \in \mathbb{R}\backslash\{0\}$, $k \geq 0$, we have
    \BS\begin{align}
        \alpha w_k &= \frac{{\rm sgn}(\alpha)}{N} \bm{1}^{\rm T} \big[(\lambda^\dag\bm{1}-\bm{\lambda_B}) \circ \bm{s}_\lambda^{\circ k} \circ \mathrm{e}^{\circ \log|\alpha|\bm{1}+k\bm{\lambda}_{B}^{\log}} \big], 
    \end{align}
    where 
    \begin{align}
        \bm{s}_{\lambda} \equiv {\rm sgn}(\bm{\lambda}_{B}),\ \bm{\lambda}_{B}^{\log} \equiv \log^{\circ}|\bm{\lambda}_{B}|.
    \end{align}
    \ES
\end{lemma}
\begin{IEEEproof}
    Recall $w_k = \lambda^\dag b_k - b_{k+1}$, where $b_k = \tfrac{1}{N}{\rm tr}\{\bm{B}^k\}$. 
    \begin{align}
    \alpha w_k &= \tfrac{1}{N} \alpha \big(\lambda^\dag{\rm tr}\{\bm{B}^k\} - {\rm tr}\{\bm{B}^{k+1}\}\big) \nonumber \\
    &= \tfrac{1}{N} \alpha \bm{1}^{\rm T} \big[(\lambda^\dag\bm{1}-\bm{\lambda}_{B})   \circ \bm{\lambda_B}^{\circ k} \big]. \label{Eqn:L3_p1}
    \end{align}
    It is clear that 
    \BS\label{Eqn:L3_p2}\begin{align}
        \alpha &= {\rm sgn}(\alpha) \mathrm{e}^{\log|\alpha|}, \\
        \bm{\lambda_B}^{\circ k} &= \bm{s}_{\lambda}^{\circ k}  \circ \mathrm{e}^{\circ k\bm{\lambda}_{B}^{\log}}.
    \end{align}\ES
    We have completed the proof by substituting (\ref{Eqn:L3_p2}) into (\ref{Eqn:L3_p1}).
\end{IEEEproof}
Since $\vartheta_{t,i} = \xi_i \prod_{j=i+1}^t\theta_j \neq 0$, we have
\begin{align}
    {\rm sgn}(\vartheta_{t,i}) = {\rm sgn}(\xi_i),\ \log|\vartheta_{t,i}| = \log|\xi_i| + \textstyle\sum_{j=i+1}^t \log|\theta_j|.
\end{align}
Thus, we can apply Lemma \ref{Lemma:oa_1} to calculate the product terms like $\vartheta_{t,i}w_{t-i}$ and $\vartheta_{t,i}\vartheta_{t,j}w_{2t-i-j}$. However, the complexity of calculating each $\vartheta w$ via Lemma \ref{Lemma:oa_1} is $\mathcal{O}(M)$. Over $T$ iterations, the algorithm needs to calculate $\mathcal{O}(T^3)$ such terms. This results in a total complexity of $\mathcal{O}(MT^3)$, which is relatively high.

To reduce this complexity, we introduce
\begin{align}\label{Eqn:rchi}
    \rchi_{k} \equiv \theta_0^k w_k,
\end{align}
where $\theta_0 = (\lambda^\dag + \sigma^2)^{-1}$. Then, any product term $\vartheta w_k$ can be computed indirectly as
\begin{align}\label{Eqn:bw}
    \vartheta w_k = {\rm sgn}(\vartheta)\mathrm{e}^{\log|\vartheta|-k\log\theta_0}\rchi_k.
\end{align}
In the $t$-th iteration, we only need to calculate $\rchi_{2t-1}$ and $\rchi_{2t}$ via Lemma \ref{Lemma:oa_1}, reducing the overall complexity from $\mathcal{O}(MT^3)$ to $\mathcal{O}(MT+T^3)$. Furthermore, the following Lemma \ref{Lemma:oa_2} establishes a bound on the magnitude of $\rchi_{k}$, ensuring its computation is free from the risk of overflow.
\begin{lemma}\label{Lemma:oa_2}
    For any $k\geq 0$, we have
    \begin{align}
        |\rchi_{k}| \leq \delta (\lambda^\dag + \theta_0^{-1}),
    \end{align}
    where $\delta = M/N$.
\end{lemma}
\begin{IEEEproof}
    Recall $w_k = \lambda^\dag b_k - b_{k+1}$, where $b_k = \tfrac{1}{N}{\rm tr}\{\bm{B}^k\}$. We have
    \begin{align}\label{Eqn:Xi}
        |\rchi_{k}| = \big| \lambda^\dag \theta_0^k b_k - \theta_0^{-1} \theta_0^{k+1} b_{k+1} \big|.
    \end{align}
    The spectral radius of $\theta_0\bm{B}$ is $(\lambda^{\dag}-\lambda_{\min}) / (\lambda^{\dag}+\sigma^2)$, which is less than 1. Thus, for $i \geq 0$,
    \begin{align}\label{Eqn:theta_b}
        -\delta \leq \theta_0^k b_k = \tfrac{1}{N}{\rm tr}\{(\theta_0\bm{B})^k\} \leq \delta.
    \end{align} 
   Both $\lambda^\dag$ and $\theta_0$ are positive. Then, following \eqref{Eqn:Xi} and \eqref{Eqn:theta_b}, we have $|\rchi_{k}| \leq \delta (\lambda^\dag + \theta_0^{-1})$.
\end{IEEEproof}
We refer to this reformulated algorithm as overflow-avoiding GD-MAMP (OA-GD-MAMP). Further details and the pseudocode of OA-GD-MAMP are presented in Appendix \ref{App:OA}. 

\subsection{OA-GD-MAMP with Unknown Eigenvalues of \texorpdfstring{$\bm{A}\bm{A}^{\rm H}$}{TEXT}}
In practice, the eigenvalues of $\bm{A}\bm{A}^{\rm H}$ are usually unknown. As described in Section \ref{Sec:pre1}, we set $\lambda^\dag$ as $(\lambda_{\max}^{\rm up} + \lambda_{\min}^{\rm low}) / 2$, where $\lambda_{\max}^{\rm up}$ and $\lambda_{\min}^{\rm low}$ are given in (\ref{Eqn:max_min}) with a suitable $\tau$. Subsequently, the exact calculation of $\rchi_k$ via Lemma \ref{Lemma:oa_1} is no longer available. To address this, the following lemma provides a method to  approximate $\rchi_k$.
\begin{lemma}\label{Lemma:oa_h}
    For $k \geq 0$, 
    \begin{align}
        \rchi_k \overset{\rm a.s.}{=} \lim_{N\to\infty} \bar{\bm{h}}_i^{\rm H} \bar{\bm{h}}_{k-i} \label{Eqn:oa_rchi}
    \end{align}
    where $i = \lceil k/2 \rceil$ and $\bar{\bm{h}}_i$ is given by a recursion
    \begin{align}
        \bar{\bm{h}}_i = \theta_0\bm{B}\bar{\bm{h}}_{i-1},
    \end{align}
    with $\bar{\bm{h}}_0 = \bm{A}\bm{h}_0$, $\bm{h}_0 \sim \mathcal{N}(\bm{0},\tfrac{1}{N}\bm{I}_{N})$.
\end{lemma}
\begin{IEEEproof}
    For $0 \leq i \leq k$, 
    \BS
    \begin{align}
        \bar{\bm{h}}_i^{\rm H} \bar{\bm{h}}_{k-i} &= \lim_{N\to\infty} \bm{h}_0^{\rm H} \bm{A}^{\rm H}(\theta_0\bm{B})^k\bm{A}\bm{h}_0 \\ 
        &\overset{\rm a.s.}{=} \lim_{N\to\infty}\tfrac{\theta_0^k}{N}{\rm tr}\{\bm{A}^{\rm H}\bm{B}^k\bm{A}\} \bm{h}_0^{\rm H} \bm{h}_0 \\
        &= \theta_0^k w_k = \rchi_k.
    \end{align}
    Hence, we have completed the proof.
    \ES
\end{IEEEproof}
Approximating $\lambda_{\max}^{\rm up}$ via  (\ref{Eqn:max_min}) introduces $\tau$ matrix-vector products, and approximating $\{\rchi_k\}$ for $k \in [2T]$ via Lemma \ref{Lemma:oa_h} introduces $2T$ matrix-vector products. Thus, OA-GD-MAMP with unknown eigenvalues of $\bm{A}\bm{A}^{\rm H}$ requires a total of $5T + \tau$ matrix-vector products. In the subsequent sections, when comparing the number of matrix-vector products, we omit $\tau$ since it is typically set to be a small constant (e.g. $\tau \leq 20$).
%\footnote{When the eigenvalues of $\bm{A}\bm{A}^{\rm H}$ are unknown, GD-MAMP approximates $\{w_k, k<2T\}$ using the similar method in Lemma \ref{Lemma:oa_h}. In other words, OA-GD-MAMP introduces no additional matrix-vector products compared to GD-MAMP.}. If $T$ is large, we can consider $\tau \ll T$. 

\section{Variants of GD-MAMP to Reduce Complexity}\label{Sec:CR_GD}
In this section, we first propose a complexity-reduced GD-MAMP, which requires only two matrix-vector products per iteration, a $33\%$ reduction compared to GD-MAMP, without a noticeable degradation in convergence speed measured by iteration count. Second, we introduce a method to perform orthogonalization in the MLE using partial-length memory terms. This reduces the space complexity of storing memory terms. In addition, when the eigenvalues of $\bm{A}\bm{A}^{\rm H}$ are unknown, both variants need to estimate fewer terms of $\rchi_k$ by Lemma \ref{Lemma:oa_h}, thus further reducing the number of matrix-vector products.

\subsection{Complexity-Reduced GD-MAMP}\label{SSec:CR}
As shown in Section \ref{SSec:SGD}, GD-MAMP requires three matrix-vector products in each iteration. The two matrix-vector products from (\ref{Eqn:SGD_MAMP_a}) and (\ref{Eqn:SGD_MAMP_b}) are clearly unavoidable. However, the one arising from computing $\bm{z}_t=\bm{y}-\bm{A}\phi_{t-1}(\bm{r}_{t-1})$ is only needed when estimating the covariances $v_{t, i}^{\phi}$ in $\bm{V}_t^{\phi}$ by (\ref{Eqn:v_phi}). Motivated by this observation, our key idea is to modify the original GD-MAMP to eliminate the dependency on $\bm{V}_t^{\phi}$, thereby avoiding computing $\bm{z}_t$.

\textbf{\emph{CR-GD-MAMP Algorithm:}} Starting with $t=1$, $\bm{u}_{0} = \bm{0}_{M}$ and $\bm{x}_1 = {\rm E}\{\bm{x}\} = \bm{0}_{N}$,
\BS\label{Eqn:MLE_CR}
\begin{alignat}{2}
    {\rm MLE:} && \bm{u}_{t} &= \theta_t\lambda^\dag\bm{u}_{t-1} + \hat{\xi}_t\bm{y} - \bm{A}(\theta_t\hat{\bm{r}}_{t-1} + \hat{\xi}_t\bm{x}_t), \\
    && \hat{\bm{r}}_{t} &= \bm{A}^{\mr H}\bm{u}_{t} \\
    && \bm{r}_t &= \gamma_t(\bm{X}_t) = \tfrac{1}{{\varepsilon}^\gamma_t}\big(\hat{\bm{r}}_{t} + \textstyle\sum_{i=1}^t  p_{t, i}\bm{x}_i \big), \\
    && \tilde{\bm{r}}_t &= \big[\bm{r}_1, \cdots\!, \bm{r}_t \big] \cdot \tilde{\scaleto{\bm{\zeta}}{8pt}}_t,
    \label{Eqn:MLE_CR_d}\\
    {\rm NLE:} &&\ \bm{x}_{t+1} &= \phi_t(\tilde{\bm{r}}_t).
\end{alignat}\ES
Compared to GD-MAMP, the differences are as follows: 
\begin{enumerate}
    \item The error covariance matrix $\bm{V}_t^{\phi}$ for $\bm{x}_1, \cdots\!, \bm{x}_t$ is no longer required. Instead, we only obtain the variance $v_{t,t}^{\phi}$ of $\bm{x}_{t}$ as in OAMP/VAMP: $v_{t,t}^{\phi} = 1$, and for $t \geq 2$, 
    \begin{align}
        v_{t,t}^{\phi} = (1 / v_{t-1}^{\rm post} - 1 / v_{t-1,t-1}^{\tilde{\gamma}})^{-1},
    \end{align}
    where $v_{t-1}^{\rm post} = {\rm Var}\big\{x|x+w, w \sim \mathcal{CN}(0, v_{t-1,t-1})\big\}$ is the posterior variance, and $v_{t-1,t-1}^{\tilde{\gamma}}$ is the variance of $\tilde{\bm{r}}_{t-1}$ sent by the MLE. 
    \item The optimized weight parameter $\xi_t$ in (\ref{Eqn:xi}) cannot be used since $\bm{V}_t^{\phi}$ is unavailable. We replace it with a heuristically selected parameter: 
    \begin{align}
        \hat{\xi}_t = 1 / (v_{t,t}^{\phi} + \sigma^2).
    \end{align}
    Note that this $\hat{\xi}_t$ does not affect the convergence of the algorithm, but only its convergence speed \cite{liu2022memory}. While $\hat{\xi}_t$ is suboptimal, our simulation results in Section \ref{Sec:sim} show that it does not degrade the convergence speed notably.
    \item The variance $v_{t,t}^{\gamma}$ of $\bm{r}_t$ cannot be obtained from \eqref{Eqn:v_gam} since $\bm{V}_t^{\phi}$ is unavailable. Instead, we estimate $v_{t,t}^{\gamma}$ by
    \begin{align}
        v_{t, t}^{\gamma} \overset{\rm a.s.}{=} \lim_{N\to\infty}\tfrac{1}{N}\| \bm{r}_t - \bm{x}_t \|^2 - v_{t,t}^{\phi}.
    \end{align}
    This technique was used in \cite{skuratovs2022compressed}. 
    \item Inspired by \cite{liu2021sufficient, liu2022sufficient}, we perform a damping at the MLE, rather than the NLE, to improve the convergence rate. Analogous to \eqref{Eqn:xi_2}, the optimal damping vector is given by
    \begin{align}\label{Eqn:zeta_opt}
        \scaleto{\bm{\zeta}}{8pt}_t^{\rm op} = \frac{[\bm{V}_{t}^{\gamma}]^{-1} \bm{1}}{\bm{1}^{\rm T} [\bm{V}_{t}^{\gamma}]^{-1}\bm{1}},
    \end{align}
    where $\bm{V}_t^{\gamma} \equiv [v_{i,j}^{\gamma}]_{t \times t}$ denotes the error covariance matrix for $\bm{r}_1, \cdots\!, \bm{r}_t$. However, it is hard to estimate the full $\bm{V}_t^{\gamma}$. We can only estimate the real part of $\bm{V}_t^{\gamma}$ by the following lemma.
    \begin{lemma}\label{Lem:CR_cov}
        For $i \in [t]$, we have 
        \begin{align}
            \Re(v_{t, i}^{\gamma}) \overset{\rm a.s.}{=} \lim_{N\to\infty} \tfrac{1}{2}\big(v_{t, t}^{\gamma} + v_{i, i}^{\gamma} - \tfrac{1}{N}\|\bm{r}_{t} - \bm{r}_i\|^2 \big). \label{Eqn:real_v_g}
        \end{align}
    \end{lemma}
    \begin{IEEEproof}
    Let $\bm{g}_t = \bm{r}_t - \bm{x}$. We have
    \begin{align}
        & \lim_{N\to\infty} \tfrac{1}{N}\|\bm{r}_{t} - \bm{r}_i\|^2 \nonumber \\
        & = \lim_{N\to\infty} \tfrac{1}{N}\|\bm{g}_{t} - \bm{g}_i\|^2 \nonumber \\
        & = \lim_{N\to\infty} \tfrac{1}{N}(\bm{g}_t^{\rm H}\bm{g}_t + \bm{g}_i^{\rm H}\bm{g}_i - \bm{g}_i^{\rm H}\bm{g}_t - \bm{g}_t^{\rm H}\bm{g}_i) \nonumber \\
        &\overset{\rm a.s.}{=} v_{t, t}^{\gamma} + v_{i, i}^{\gamma} - 2\Re(v_{t, i}^{\gamma}).
    \end{align}
    Thus, we have finished the proof.
    \end{IEEEproof}
    In this case, we replace $\bm{V}_t^{\gamma}$ with $\Re(\bm{V}_t^{\gamma})$ in \eqref{Eqn:zeta_opt}, leading to the damping vector
    \begin{align}\label{Eqn:CR_da}
        \tilde{\scaleto{\bm{\zeta}}{8pt}}_t = \frac{[\Re(\bm{V}_{t}^{\gamma})]^{-1} \bm{1}}{\bm{1}^{\rm T} [\Re(\bm{V}_{t}^{\gamma})]^{-1}\bm{1}},
    \end{align}
    where $\Re(\bm{V}_t^{\gamma})$ is assumed to be invertible\footnote{If $\Re(V_t^{\gamma})$ is singular or ill-conditioned, we employ the back-off damping in \cite{liu2022memory}, or add a small amount of diagonal loading.}. 
    %Subsequently, the error covariance matrix $\bm{V}_t^{\tilde{\gamma}} \equiv [v_{i, j}^{\tilde{\gamma}}]_{t \times t}$ for the damped estimates $\tilde{\bm{r}}_1, \cdots\!, \tilde{\bm{r}}_t$ is given by computing its $t$-th row/column:
    % \begin{align}
    %     v^{\tilde{\gamma}}_{t, i} = v^{\tilde{\gamma}}_{i, t} = 1 \;/\; \bm{1}^{\rm T} [\Re(\bm{V}_t^{\gamma})]^{-1}\bm{1},\ \forall i \in [t].
    % \end{align}
    For real-valued systems, $\tilde{\scaleto{\bm{\zeta}}{8pt}}_t$ in (\ref{Eqn:CR_da}) is optimal since $\Re(\bm{V}_{t}^{\gamma}) = \bm{V}_{t}^{\gamma}$. For complex-valued systems, while this replacement is an approximation, the following lemma shows that it remains optimal if the damping vector is constrained to be real.
    \begin{lemma}\label{Lem:CR2}
        Suppose that $\bm{V}_t^{\gamma} \in \mathbb{C}^{t \times t}$ is invertible. If the damping vector $\scaleto{\bm{\zeta}}{8pt}$ is real-valued and satisfies $\bm{1}^{\rm T} \scaleto{\bm{\zeta}}{8pt} = 1$, then $ \tilde{\scaleto{\bm{\zeta}}{8pt}}_t$ in (\ref{Eqn:CR_da}) minimizes the variance of $\bm{r}(\scaleto{\bm{\zeta}}{8pt}) = [\bm{r}_1, \cdots\!, \bm{r}_t] \cdot \scaleto{\bm{\zeta}}{8pt}$.
    \end{lemma}
    \begin{IEEEproof}
        For any $\scaleto{\bm{\zeta}}{8pt}$ satisfying $\bm{1}^{\rm H}\scaleto{\bm{\zeta}}{8pt}=1$, the variance of $\bm{r}(\scaleto{\bm{\zeta}}{8pt})$ is given by \cite{liu2022memory}:
        \begin{align}
            v_{t, t}(\scaleto{\bm{\zeta}}{8pt}) = \scaleto{\bm{\zeta}}{8pt}^{\rm H}\, \bm{V}_t^{\gamma} \,\scaleto{\bm{\zeta}}{8pt}.
        \end{align}
        We have $\scaleto{\bm{\zeta}}{8pt}^{\rm H}\,\bm{V}_t^{\gamma} \,\scaleto{\bm{\zeta}}{8pt} = \scaleto{\bm{\zeta}}{8pt}^{\rm T}\,\Re(\bm{V}_t^{\gamma}) \,\scaleto{\bm{\zeta}}{8pt}$ when $\scaleto{\bm{\zeta}}{8pt}$ is real-valued. It is straightforward to verify that $ \tilde{\scaleto{\bm{\zeta}}{8pt}}_t$ in (\ref{Eqn:CR_da}) is the solution of
        \BS
        \begin{align}
            &\min_{\scaleto{\bm{\zeta}}{8pt}}\ \; v_{t, t}(\scaleto{\bm{\zeta}}{8pt}) = \scaleto{\bm{\zeta}}{8pt}^{\rm H}\,\Re(\bm{V}_t^{\gamma})\,\scaleto{\bm{\zeta}}{8pt}, \\
            &\;{\rm s.t.}\ \; \bm{1}^{\rm T}\scaleto{\bm{\zeta}}{8pt} = 1,\ \scaleto{\bm{\zeta}}{8pt} \in \mathbb{R}^t.
        \end{align}
        \ES
        Thus, we have finished the proof.
    \end{IEEEproof}
    Similar to GD-MAMP, we typically set the maximum damping length as $L=2$ or $L=3$ to enhance the robustness. In addition, based on our experience, for systems of moderate scale (e.g. $N < 1000$), the estimation of $\Re(v_{t, i}^{\gamma})$ in \eqref{Eqn:real_v_g} can become inaccurate, leading to algorithmic instability. In this case, we can use a simple heuristic damping as
    \begin{align}
        \scaleto{\bm{\zeta}}{8pt}_t^{\rm he} = \Big[\frac{v_{t, t}^{\gamma}}{{v_{t-1, t-1}^{\gamma}+v_{t, t}^{\gamma}}},\  \frac{v_{t-1, t-1}^{\gamma}}{{v_{t-1, t-1}^{\gamma}+v_{t, t}^{\gamma}}}\Big]
    \end{align}
\end{enumerate}
%Recently, a sufficient-statistic GD-MAMP (SS-GD-MAMP) was proposed by moving the optimal damping to the MLE in \cite{liu2022sufficient}. For SS-GD-MAMP, the damping at NLE is not necessary (cannot improve the convergence speed). In addition, the simulation results presented in \cite{liu2022sufficient} also demonstrate that SS-GD-MAMP convergences quicker than GD-MAMP. All these results above indicate that damping at MLE is a better choice. That is why we adopt this approach when designing CR-GD-MAMP.

The proposed CR-GD-MAMP requires two matrix-vector products per iteration, while GD-MAMP requires three. The complexity advantage is more pronounced when the eigenvalues of $\bm{A}\bm{A}^{\rm H}$ are unknown. In this case, CR-GD-MAMP only requires estimating $\rchi_1, \cdots\!, \rchi_T$ by Lemma \ref{Lemma:oa_h} over $T$ iterations, since computing \eqref{Eqn:v_gam} is unnecessary. As a result, the total number of matrix-vector products is reduced to nearly $3T$. This count consists of $2T$ products for the main iterations, and $2\lceil0.5T\rceil \approx T$ for the estimation of $\rchi_1, \cdots\!, \rchi_T$. As shown in Table \ref{table:CR}, CR-GD-MAMP reduces the required matrix-vector products by $1/3$ (from $3T$ to $2T$) and $2/5$ (from $5T$ to $3T$) for known and unknown eigenvalues of $\bm{A}^{\mr{H}}\bm{A}$, respectively.

\begin{table}[htb!]
\renewcommand{\arraystretch}{1.5} 
\centering \footnotesize \setlength{\tabcolsep}{1mm}{
\caption{Total Number of matrix-vector products in GD-MAMP and CR-GD-MAMP \\ ($T$: number of iterations)}
\label{table:CR}
\begin{tabular}{|c|c|c|}
\hline
\multirow{2}{*}{\quad Algorithms\quad} & \multicolumn{2}{c|}{Total number of matrix-vector products} \\
\cline{2-3}
&\ Known eigenvalues of $\bm{A}\bm{A}^{\rm H}$ \ & \ Unknown eigenvalues of $\bm{A}\bm{A}^{\rm H}$ \  \\
\hline
GD-MAMP \cite{liu2022memory} & $3T$ & $5T$ \\
\hline
CR-GD-MAMP & $2T$ & $3T$ \\
\hline
\end{tabular}} 
\end{table}

\subsection{GD-MAMP with Partial-Length Memory Terms}\label{SSec:PMLE}
The time complexity of GD-MAMP is primarily dominated by the complexity $\mathcal{O}(MNT)$ arising from matrix-vector products. However, for a large $T$, the secondary complexity $\mathcal{O}(NT^2)$, stemming from the orthogonalization step $\big\{\textstyle\sum_{i=1}^t  p_{t, i}\bm{x}_i\}_{t=1}^T$, is also non-negligible. More importantly, the space complexity of storing memory terms $\bm{x}_1, \cdots\!, \bm{x}_t$ is $\mathcal{O}(NT)$, which is relatively high when storage resources are limited. 

Let us analyze the orthogonalization. For a fixed $k$,
\begin{align}
    |p_{t,k}| \propto {\rm tr}\big\{\textstyle\prod_{\tau=k+1}^t\theta_\tau\bm{A}^{\rm H}\bm{B}^{t-k}\bm{A}\big\}.
\end{align}
Since $\theta_1 > \cdots > \theta_t > 0$ and $\rho(\theta_1\bm{B}) < 1$ \cite{liu2022memory}, $|p_{t,k}|$ decreases rapidly as the iteration count $t$ increases. If $t-k$ is sufficiently large, $p_{t,k} \to 0$, meaning that the influence of $\bm{x}_1, \cdots, \bm{x}_k$ for $\hat{\bm{r}}_t$ is negligible. Therefore, we consider replacing the MLE with a partial MLE, i.e., (\ref{Eqn:SGD_MAMP_c}) is modified to 
\begin{align}\label{Eqn:PM_MLE}
    \bm{r}_t = \tfrac{1}{{\varepsilon}^\gamma_t}\big(\hat{\bm{r}}_{t} + \textstyle\sum_{i=l_t}^t  p_{t, i}\bm{x}_i \big),
\end{align}
where $l_t \geq 1$ is a length parameter. In particular, the variant degenerates to the original GD-MAMP when $l_t = 1$. Corresponding to (\ref{Eqn:PM_MLE}), for $1 \leq  t' \leq t$,
\begin{align}
    v^{\gamma}_{t,t'} \approx \tfrac{1}{\varepsilon^\gamma_{t}\varepsilon^\gamma_{t'}} \textstyle{\sum}_{i=l_t}^{t}\textstyle{\sum}_{j=l_{t'}}^{t'}  \vartheta_{t, i}\vartheta_{t'\!,j}\big[\sigma^2  w_{t+t'-i-j} + v^{\bar{\phi}}_{i,j}\bar{ w}_{t-i,t'\!-j}\big].
\end{align}
In addition, $c_{t,0}, c_{t,2}, c_{t,3}$ are changed to
\BS\begin{align}
    c_{t,0} &= \textstyle\sum_{i=l_t}^{t-1} \vartheta_{t, i} w_{t-i} / w_0,\\
    c_{t,2} &= -\textstyle{\sum}_{i=l_t}^{t-1} \vartheta_{t, i} \big(\sigma^2 w_{t-i}+ \Re(v_{{t},{i}}^{\bar{\phi}})\bar{ w}_{0,t-i}\big), \\
    c_{t,3} &= \textstyle{\sum}_{i=l_t}^{t-1}\textstyle{\sum}_{j=l_t}^{t-1} \vartheta_{t, i}\vartheta_{t,j}\big(\sigma^2 w_{2t-i-j} + v^{\bar{\phi}}_{i,j}\bar{ w}_{t-i,t-j}\big),
\end{align}\ES
For simplicity, we can introduce a heuristic $l$ and let
\begin{align}
    l_t  = \max(1, t-l+1).
\end{align}
In practice, we can set $l$ according to the state evolution (SE). As a result, the time and space complexity due to the orthogonalization are reduced to $\mathcal{O}(NTl)$ and $\mathcal{O}(Nl)$ respectively. Moreover, when the eigenvalues of $\bm{A}\bm{A}^{\rm H}$ are unknown, OA-GD-MAMP with the partial MLE only requires approximating $\{\rchi_k, k<2l\}$. Thus, the number of matrix-vector products is reduced from $5T$ to $3T+2l$.

\section{General Formulation of MAMP Algorithms and Reformulated warm-started conjugate gradient VAMP}\label{Sec:CG}
In this section, we first develop a general formulation for designing MAMP algorithms based on gradient methods. The warm-started conjugate gradient VAMP (WS-CG-VAMP) \cite{skuratovs2022compressed} can be recovered as a special case of this formulation by applying the CG method. Then, we show that the computation of the orthogonalization parameters is susceptible to catastrophic cancellation, which may lead to numerical instability in finite-precision arithmetic. This analysis explains the finite-precision instability observed in WS-CG-VAMP. Next, we derive an equivalent reformulation of WS-CG-VAMP that eliminates redundant computations, yielding WS-CG-VAMP(r), which reduces the number of matrix-vector products by up to $50\%$. Finally, we compare the convergence speed of GD-MAMP and WS-CG-VAMP(r) in terms of the number of matrix-vector products. Our simulation results show that GD-MAMP converges faster than WS-CG-VAMP(r) when the condition number $\kappa(\bm{A})$ is small. When $\kappa(\bm{A})$ is large, WS-CG-VAMP(r) converges faster than GD-MAMP under high-precision arithmetic. However, under IEEE double-precision arithmetic, WS-CG-VAMP(r) diverges due to catastrophic cancellation. 

\subsection{General Formulation for Designing MAMP Algorithms} \label{Sec:GM_MAMP}
A general gradient-based formulation for designing MAMP algorithms is presented as follows: starting with $\bm{x}_1^{\phi} = \bm{0}_N$ and the outer iteration count $t = 1$,
\BS \label{Eqn:GF}
\begin{alignat}{2}
    {\rm Memory\ LE:} \quad && \quad \bm{x}_t^{\gamma} &= \tfrac{1}{\varepsilon^{\gamma}_{t}}\big[\bm{A}^{\mr H}\bm{u}^{(I)}_{t} + \textstyle\sum_{k=1}^t \omega^{(I)}_{t,k}\bm{x}_k^{\phi} \big], \label{Eqn:GF_MLE} \\
    {\rm NLE:} \quad && \bm{x}_{t+1}^{\phi} &= \phi_t(\bm{x}_t^{\gamma}). \label{Eqn:GF_NLE}
\end{alignat}
\ES
where $\phi_t(\cdot)$ is the same separable and
Lipschitz-continuous NLE as in OAMP/VAMP \cite{ma2017orthogonal, rangan2019vector}. Depending on the design requirements, one may use either the optimized damping in \eqref{Eqn:xi_2} or a heuristic damping. We omit damping here for simplicity. The details of the memory LE (MLE) are given below.

Let $\bm{z}_t = \bm{y} - \bm{A}\bm{x}_{t}^{\phi}$ and $\bm{W}_t = \sigma^2 \bm{I} + v_{t, t}^{\phi} \bm{A}\bm{A}^{\rm H}$, where $v_{t, t}^{\phi}$ is the error variance of $\bm{x}_t^{\phi}$. Since $\bm{W}_t$ is Hermitian positive definite, $\bm{W}_t^{-1}\bm{z}_t$ is the minimizer of the strongly convex quadratic function $F_t(\bm{u}) = \frac{1}{2}\bm{u}^{\rm H}\bm{W}_t\bm{u} - \Re(\bm{z}_t^{\rm H}\bm{u})$. The vector $\bm{u}^{(I)}_{t}$ approximates $\bm{W}_t^{-1}\bm{z}_t$, obtained after $I$ inner iterations of the following generic gradient method: for $0 \leq i < I$,
\BS\label{Eqn:gen_grad}
\begin{align}
    \Delta\bm{u}_t^{(i)} &= \bm{u}_t^{(i)} - \bm{u}_t^{(i-1)}, \label{Eqn:du}\\
    \hat{\bm{u}}_t^{(i)} &= \bm{u}_t^{(i)} + \eta_t^{(i)} \Delta\bm{u}_t^{(i)},  \\
    \hat{\bm{r}}_t^{(i)} &= \bm{z}_t - \bm{W}_t  \hat{\bm{u}}_t^{(i)}, \\
    \bm{p}_t^{(i+1)} &= \hat{\bm{r}}_t^{(i)} + \beta_t^{(i)} \bm{p}_t^{(i)}, \\
    \bm{u}_{t}^{(i+1)} &= \bm{u}_t^{(i)} + \theta_t^{(i)} \Delta\bm{u}_t^{(i)} + \alpha_t^{(i)}\bm{p}_t^{(i+1)}, \label{Eqn:u_t}
\end{align}
\ES
where $\eta_t^{(i)}$ is the look-ahead coefficient controlling the residual-evaluation point, $\beta_t^{(i)}$ is the direction coefficient weighting the previous search direction, $\theta_t^{(i)}$ is the momentum coefficient weighting the previous displacement, $\alpha_t^{(i)}$ is the step size along the current search direction. The initialization of $\bm{u}_t^{(0)}$, $\bm{u}_t^{(-1)}$ and $\bm{p}_t^{(0)}$ is given by 
\begin{align}
    \bm{u}_t^{(0)} = \bm{u}_{t-1}^{(I)},\ \,\bm{u}_t^{(-1)} = \bm{u}_{t-1}^{(I-1)},\ \,\bm{p}_t^{(0)} = \bm{p}_{t-1}^{(I)}, \label{Eqn:ini}
\end{align}
with $\bm{u}_1^{(0)} = \bm{u}_{1}^{(-1)} = \bm{p}_{1}^{(0)} = \bm{0}_M$.

The iteration in \eqref{Eqn:gen_grad} recovers several classical gradient methods, including
\begin{itemize}
    \item Gradient descent (GD): $\eta_t^{(i)}=\beta_t^{(i)}=\theta_t^{(i)}=0$. A typical choice is
    \begin{align}
        \alpha_t^{(i)} = \big(\sigma^2+v_{t, t}^{\phi}\lambda^\dag\big)^{-1},
    \end{align}
    where $\lambda^{\dag} = (\lambda_{\max} + \lambda_{\min}) / 2$ denotes the average of extreme eigenvalues of $\bm{A}\bm{A}^{\rm H}$. 
    \item Nesterov accelerated gradient (NAG): 
    $\beta_t^{(i)}=0$, and
    \begin{align}
        \alpha_t^{(i)} &= \big(\sigma^2 + v_{t, t}^{\phi}\lambda_{\max}\big)^{-1}, \\
        \eta_t^{(i)} &= \theta_t^{(i)} = \frac{\sqrt{\tilde{\kappa}}-1}{\sqrt{\tilde{\kappa}}+1},
    \end{align}
    where $\tilde{\kappa} = (\sigma^2 + v_{t, t}^{\phi}\lambda_{\max}) / (\sigma^2 + v_{t, t}^{\phi}\lambda_{\min})$.
    \item Conjugate gradient (CG): $\eta_t^{(i)}=\theta_t^{(i)}= 0$, 
    \begin{align}
        \alpha_{t}^{(i)} = \frac{\|\bm{r}_t^{(i)}\|^2} {\big\langle\bm{p}_t^{(i+1)},\bm{W}_t\bm{p}_t^{(i+1)}\big\rangle}, \ \;
        \beta_t^{(i)} = \frac{\|\bm{r}_t^{(i)}\|^2}{\|\bm{r}_{t}^{(i-1)}\|^2},
    \end{align}
    where $\beta_1^{(0)} = 0$ and $\beta_t^{(0)}=\beta_{t-1}^{(I)}$ for $t \geq 2$.
\end{itemize}
Different choices of these coefficients recover different gradient methods and hence yield different MAMP algorithms.  In the following lemma, we show that $\bm{u}_t^{(i)}$ and $\bm{p}_t^{(i)}$ belong to Krylov-type spans generated by $\{\bm{z}_k\}_{k=1}^t$, with coefficient matrices given by polynomials in $\bm{A}\bm{A}^{\mr{H}}$.
\begin{lemma}\label{Lem:gen_MAMP}
    For $t \geq 1$, $I \geq 1$ and $0 \leq i \leq I$, we have
    \BS \label{Eqn:Lem8}
    \begin{align}
        \bm{p}_t^{(i)} &= \textstyle\sum_{k=1}^t \bm{C}_{t,k}^{(i)}\bm{z}_k,\\
        \bm{u}_t^{(i)} &= \textstyle\sum_{k=1}^t \bm{D}_{t,k}^{(i)}\bm{z}_k, \label{Eqn:Lem8_b}\\
        \bm{C}_{t,k}^{(i)} &= \textstyle\sum_{j=0}^{(t-k)I+i-1} c_{t,k,j}^{(i)} (\bm{A}\bm{A}^{\rm H})^j,\\ 
        \bm{D}_{t,k}^{(i)} &= \textstyle\sum_{j=0}^{(t-k)I+i-1} d_{t,k,j}^{(i)} (\bm{A}\bm{A}^{\rm H})^j,
    \end{align}
    where $\{c_{t,k,j}^{(i)}\}, \{d_{t,k,j}^{(i)}\}$ are some power-basis coefficients. An empty sum is understood as zero, i.e., $\bm{C}_{t, t}^{(0)} = \bm{D}_{t, t}^{(0)} = \bm{0}_{M \times M}$.
    \ES
\end{lemma}
\begin{IEEEproof}
    See Appendix \ref{App:gen_MAMP}.
\end{IEEEproof}
While closed-form expressions for $c_{t,k,j}^{(i)}$ and $d_{t,k,j}^{(i)}$ are difficult to derive (except for standard gradient descent with $\eta_t^{(i)}=\beta_t^{(i)}=\theta_t^{(i)}=0$), we provide recursive constructions for ${c_{t,k,j}^{(i)}}$ and ${d_{t,k,j}^{(i)}}$ in the following lemma.
\begin{lemma}\label{Lem:c_and_d}
For $1 \leq k \leq t$ and $0 \leq i \leq I$, define the coefficient vectors of $\bm{C}_{t, k}^{(i)}$ and $\bm{D}_{t, k}^{(i)}$ in \eqref{Eqn:Lem8}, respectively, as
\BS\begin{align}
    \bm{c}_{t,k}^{(i)} &= \big[c_{t, k, 0}^{(i)}, \cdots, c_{t, k, (t-k)I+i-1}^{(i)}\big], \\
    \bm{d}_{t,k}^{(i)} &= \big[d_{t, k, 0}^{(i)}, \cdots, d_{t, k, (t-k)I+i-1}^{(i)}\big].
\end{align}\ES
For $m \geq 0$, let $\bm{0}_m^{\mr row}$ denote the $m$-dimensional zero row vector. Thus, $\bm{c}_{t,t}^{(0)}=\bm{d}_{t,t}^{(0)}=\bm{0}_0^{\mr row}$, where $\bm{0}_0^{\mr row}$ is the empty vector. For $\bm{a} \in \mathbb{C}^{1 \times m}$, define the left and right zero-padding operators as
\begin{align}
    \mathcal{L}(\bm{a}) = [0, \bm{a}],\ \; \mathcal{R}(\bm{a}) = [\bm{a}, 0].
\end{align}
Then, for $0 \leq i < I$, we have
\BS\label{Eqn:recursion}\begin{align}
    \bm{s}_{t,k}^{(i)} &= 
    \begin{cases}
        \mathcal{R}\big(\bm{d}_{t-1,k}^{(I-1)}\big),
            &\ \; i=0,\ k<t,\\
        \bm{0}_0^{\mr row},
            &\ \; i=0,\ k=t,\\
        \mathcal{R}\big(\bm{d}_{t,k}^{(i-1)}\big),
            &\ \; 1\leq i<I.
    \end{cases} 
    \displaybreak[4] \\
    \hat{\bm{d}}_{t,k}^{(i)} &= \bm{d}_{t,k}^{(i)} + \eta_t^{(i)} \big(\bm{d}_{t,k}^{(i)} - \bm{s}_{t,k}^{(i)}\big), \\
    \bm{c}_{t,k}^{(i+1)} &= \big[\delta_{t,k}, \mathbf{0}_{(t-k)I+i}^{\mr row}\big] + \mathcal{R}\Big(\beta_t^{(i)}\bm{c}_{t,k}^{(i)} - \sigma^2 \hat{\bm{d}}_{t,k}^{(i)}\Big) - v_{t,t}^{\phi}\mathcal{L}\Big(\hat{\bm{d}}_{t,k}^{(i)}\Big), \\
    \bm{d}_{t,k}^{(i+1)} &= \mathcal{R}\Big(\bm{d}_{t,k}^{(i)} + \theta_t^{(i)}\big(\bm{d}_{t,k}^{(i)} - \bm{s}_{t,k}^{(i)}\big)\Big) + \alpha_t^{(i)}\bm{c}_{t,k}^{(i+1)},
\end{align}\ES
where $\delta_{t, k} = 0$, $\bm{c}_{t,k}^{(0)} = \bm{c}_{t-1, k}^{(I)}$ and $\bm{d}_{t,k}^{(0)} = \bm{d}_{t-1, k}^{(I)}$ for $k < t$, and $\delta_{t, t} = 1$.
\end{lemma}
\begin{IEEEproof}
Under the polynomial representations in Lemma \ref{Lem:gen_MAMP}, the right zero-padding operator $\mathcal{R}(\cdot)$ preserves a polynomial in $\bm{A}\bm{A}^{\rm H}$, whereas the left zero-padding operator $\mathcal{L}(\cdot)$ corresponds to multiplication by $\bm{A}\bm{A}^{\rm H}$. Applying these facts to \eqref{Eqn:gen_grad}, together with \eqref{Eqn:ini}, we complete the proof.
\end{IEEEproof}

By Lemma~\ref{Lem:gen_MAMP}, we can rewrite $\bm{A}^{\rm H}\bm{u}_t^{(I)}$ in \eqref{Eqn:GF_MLE} as
\begin{align}
    \bm{A}^{\rm H}\bm{u}_t^{(I)} &= \bm{A}^{\rm H}\textstyle\sum_{k=1}^{t} \bm{D}_{t,k}^{(I)}
    (\bm{y}-\bm{A}\bm{x}_k^{\phi}) \notag \\
    &= \underbrace{\big(\textstyle\sum_{k=1}^{t}\bm{A}^{\rm H}\bm{D}_{t, k}^{(I)}\big)}_{\hat{\bm{Q}}_t}\bm{y} + \textstyle\sum_{k=1}^{t}\underbrace{-\bm{A}^{\rm H}\bm{D}_{t, k}^{(I)}\bm{A}}_{\hat{\bm{P}}_{t,k}}\bm{x}_k^{\phi}. \label{Eqn:Q_P}
\end{align}
Since $\bm{D}_{t,k}^{(I)}$ is a polynomial in $\bm{A}\bm{A}^{\rm H}$, both $\hat{\bm{Q}}_t\bm{A}$ and $\hat{\bm{P}}_{t,k}$ are polynomials in
$\bm{A}^{\rm H}\bm{A}$. Thus, Lemma~\ref{Lemma:MAMP} applies. The orthogonalization parameters $\{\omega^{(I)}_{t,k}\}$ and the normalization factor $\varepsilon^{\gamma}_{t}$ are given by
\BS\begin{align}
    \omega^{(I)}_{t,k} &= \tfrac{1}{N}{\rm tr}\big\{\bm{A}^{\rm H}\bm{D}_{t,k}^{(I)}\bm{A}\big\} \label{Eqn:gen_orth_a} \\
    &= \textstyle\sum_{j=0}^{J_{t, k}} d_{t,k,j}^{(I)} \lambda_{j+1}, \label{Eqn:gen_orth_b} \\
    \varepsilon^{\gamma}_{t} &= \textstyle\sum_{k=1}^{t} \omega^{(I)}_{t,k}, \label{Eqn:gen_ep}
\end{align}\ES
where $J_{t, k} = (t-k+1)I-1$, $\lambda_{j} = \tfrac{1}{N}{\rm tr}\{(\bm{A}\bm{A}^{\rm H})^{j}\}$, and the coefficients $\{d_{t,k,j}^{(I)}\}$ are computed recursively from Lemma~\ref{Lem:c_and_d}.

Next, for $1 \leq t' \leq t$, let $v_{t, t'}^\gamma$ denote the covariance between $\bm{x}_t^{\gamma}$ and $\bm{x}_{t'}^{\gamma}$, given by
\begin{align}\label{Eqn:gen_var}
    v_{t, t'}^\gamma = \frac{1}{(\varepsilon^{\gamma}_{t})^{*}\varepsilon^{\gamma}_{t'}}\textstyle\sum_{k_1=1}^t\sum_{k_2=1}^{t'}\sum_{j_1=0}^{J_{t, k_1}}\sum_{j_2=0}^{J_{t', k_2}} \Big[(d_{t,k_1,j_1}^{(I)})^{*}d_{t',k_2,j_2}^{(I)}
    (\sigma^2\lambda_{j_1+j_2+1} + \tilde{\lambda}_{j_1,j_2}v_{k_1,k_2}^{\phi})\Big],
\end{align}
where $\tilde{\lambda}_{j_1,j_2} \equiv \lambda_{j_1+j_2+2}-\lambda_{j_1+1}\lambda_{j_2+1}$. See Appendix \ref{App:Var} for the derivation. The variance $v_{t, t}^{\gamma}$ can also be estimated by \cite{skuratovs2022compressed}:
\begin{align}
    v_{t, t}^{\gamma} \overset{\rm a.s.}{=} \lim_{N\to\infty}\tfrac{1}{N}\| \bm{x}_t^{\gamma} - \bm{x}_t^{\phi} \|^2 - v_{t,t}^{\phi}.
\end{align}
Note that the algorithm requires only the variance $v_{t, t}^\gamma$, while the SE needs $\{v_{t, t'}^{\gamma}\}_{t'=1}^{t}$.
\begin{remark}
    The initialization in \eqref{Eqn:ini} causes $\bm{u}_t^{(I)}$ to retain dependence on $\bm{z}_1, \cdots\!, \bm{z}_{t-1}$. Consequently, the orthogonalization term in \eqref{Eqn:GF_MLE} involves all the input estimates $\bm{x}_1^{\phi}, \cdots\!, \bm{x}_t^{\phi}$, rather than only the current estimate $\bm{x}_t^{\phi}$. In contrast, if we reset $\bm{u}_t^{(0)} = \bm{u}_{t}^{(-1)} = \bm{p}_{t}^{(0)} = \bm{0}_M$ at each outer iteration, then $\bm{u}_t^{(I)}$ depends only on $\bm{z}_t$, which implies that $\omega^{(I)}_{t,k} = 0$ for $k <t$. Hence, the MLE reduces to a memoryless linear estimator. In this case, the fixed-point MSE performance is generally worse than that of OAMP/VAMP unless $I$ is large enough for $\bm{u}_t^{(I)}$ to accurately approximate $\bm{W}_t^{-1}\bm{z}_t$\cite{takeuchi2022convergence, liu2021sufficient, liu2022sufficient}. EP-CG proposed in \cite{takeuchi2017rigorous} is an example of such a memoryless MLE implemented using the conjugate gradient method.
\end{remark}

\subsection{Finite-Precision Instability in Computing Orthogonalization Parameters}\label{Sec:gen_precision}
Recall from \eqref{Eqn:gen_orth_a} that
\begin{align}
    \omega_{t,k}^{(I)} = \tfrac{1}{N}{\rm tr}\big\{\bm{A}^{\rm H}\bm{D}_{t,k}^{(I)}\bm{A}\big\}.
\end{align}
By Lemma \ref{Lem:gen_MAMP}, we know that $\bm{A}^{\rm H}\bm{D}_{t,k}^{(I)}\bm{A}$ is a polynomial in $\bm{A}^{\rm H}\bm{A}$, with the power-basis coefficient vector $\big[0, d_{t, k, 0}^{(I)}, \cdots\!, d_{t, k, J_{t, k}}^{(I)}\big]$, where $J_{t, k} = (t-k+1)I-1$. This is why 
\begin{align}
    \omega_{t,k}^{(I)} = \textstyle\sum_{j=0}^{J_{t, k}} d_{t,k,j}^{(I)} \lambda_{j+1}. \label{Eqn:omega_rp}
\end{align}
Without loss of generality, we assume that $d_{t,k,j}^{(I)} \in \mathbb{R}$. The key point is that a small magnitude of $\omega_{t,k}^{(I)}$ does not mean that each individual summand $d_{t,k,j}^{(I)}\lambda_{j+1}$ has a small magnitude. Since $\lambda_{j+1} \geq 0$, the signs of $d_{t,k,j}^{(I)}\lambda_{j+1}$ are determined by
$d_{t,k,j}^{(I)}$. For a large $J_{t, k}$, some summands
$d_{t,k,j}^{(I)}\lambda_{j+1}$ may be large and positive, while others may be large and negative, even when their exact sum $\omega_{t,k}^{(I)}$ is much smaller in magnitude. In finite-precision floating-point arithmetic, the subtraction of such nearly balanced quantities can cause a considerable loss of significant digits. More explicitly, let $\tau_{t,k,j}^{(I)} \equiv d_{t,k,j}^{(I)}\lambda_{j+1}$, and
\begin{align}
    S_{t,k}^{+} = \textstyle\sum_{\tau_{t,k,j}^{(I)}>0} \tau_{t,k,j}^{(I)},\ \;
    S_{t,k}^{-} = -\textstyle\sum_{\tau_{t,k,j}^{(I)}<0} \tau_{t,k,j}^{(I)}.
\end{align}
It is clear that $\omega_{t,k}^{(I)} = S_{t,k}^{+}-S_{t,k}^{-}$. Assume that $\omega_{t,k}^{(I)} \neq 0$. Define the ``cancellation factor'' as
\begin{align}
    \mathcal{C}_{t,k} \equiv 
    \frac{\textstyle\sum_{j=0}^{J_{t, k}}\big|\tau_{t, k, j}^{(I)}\big|}{\big|\omega_{t, k}^{(I)}\big|} = 
    \frac{S_{t,k}^{+}+S_{t,k}^{-}}{\big|S_{t,k}^{+}-S_{t,k}^{-}\big|},
    \label{Eqn:can_condition}
\end{align}
which measures the total magnitude of the summands relative to the magnitude of their final sum. Specifically, 
\begin{itemize}
    \item $\mathcal{C}_{t,k} = 1$: All the nonzero terms $\tau_{t,k,j}^{(I)}$ have the same sign. Hence, no cancellation occurs in the summation.
    \item $\mathcal{C}_{t,k} \gg 1$: The final sum $\omega_{t,k}^{(I)}$ is obtained by summing positive and negative terms that are much larger than $\omega_{t,k}^{(I)}$, indicating severe cancellation. For example, if $S_{t,k}^{+}=10^{16}+1$ and $S_{t,k}^{-}=10^{16}$, then $\mathcal{C}_{t,k} = 2 \times 10^{16}+1$.
\end{itemize}
A large $\mathcal{C}_{t,k}$ may lead to a catastrophic loss of significant digits in finite-precision arithmetic. To illustrate this more clearly, we present the following lemma.
\begin{lemma}\label{Lem:round}
    Let $\omega_{t,k}^{\rm fp}$ denote the result of computing $\omega_{t,k}^{(I)} = \textstyle\sum_{j=0}^{J_{t, k}} d_{t,k,j}^{(I)} \lambda_{j+1}$ using standard sequential floating-point arithmetic. Assume that $\omega_{t,k}^{(I)}\neq0$, that $\{d_{t,k,j}^{(I)}\}$ and $\{\lambda_{j+1}\}$ are exact, and that no overflow or underflow occurs. Then, the relative roundoff error satisfies
    \begin{align}
        \frac{\big|{\omega}_{t,k}^{\mr{fp}}-\omega_{t,k}^{(I)}\big|}{\big|\omega_{t,k}^{(I)}\big|}\leq \frac{(J_{t,k}+1) u_{\rm fp}}{1-(J_{t,k}+1)u_{\rm fp}} \mathcal{C}_{t,k},
        \label{Eqn:relative_error}
    \end{align}
    where $\mathcal{C}_{t,k}$ is defined in \eqref{Eqn:can_condition}, and $u_{\rm fp}$ denotes the unit roundoff.
\end{lemma}
\begin{IEEEproof}
    See Appendix \ref{App:proof_round}.
\end{IEEEproof}
In almost all practical cases, $(J_{t,k}+1)u_{\rm fp} \ll 1$. Thus, we have
\begin{align}
    \frac{(J_{t,k}+1) u_{\rm fp}}{1-(J_{t,k}+1)u_{\rm fp}} \mathcal{C}_{t,k} \approx (J_{t,k}+1) u_{\rm fp}\mathcal{C}_{t,k}.
\end{align}
When $(J_{t,k}+1)u_{\rm fp}\ \mathcal{C}_{t,k}$ approaches or exceeds one, the upper bound on the relative error becomes of order one. Consequently, the bound can no longer guarantee even one correct significant digit in ${\omega}_{t,k}^{\mr{fp}}$. In other words, the rounding error may be as large as the exact $\omega_{t, k}^{(I)}$ in the worst case. 

For example, consider IEEE double-precision arithmetic with $u_{\rm fp} = 2^{-53} \approx 1.11 \times 10^{-16}$, $S_{t,k}^{+}=10^{16}+1$ and $S_{t,k}^{-}=10^{16}$. We have
\begin{align}
    \omega_{t,k}^{(I)}=1,\ \; \omega_{t,k}^{\mr{fp}} = \mr{fl}(10^{16}+1) - \mr{fl}(10^{16}) = 10^{16} - 10^{16} = 0.
\end{align}
The relative error and its upper bound are
\begin{align}
    1 = \frac{\big|{\omega}_{t,k}^{\mr{fp}}-\omega_{t,k}^{(I)}\big|}{\big|\omega_{t,k}^{(I)}\big|}\leq \frac{(J_{t,k}+1) u_{\rm fp}}{1-(J_{t,k}+1)u_{\rm fp}} \mathcal{C}_{t,k} \approx 2.2(J_{t,k}+1).
\end{align}
Clearly, catastrophic cancellation occurs in this example.

The bound in \eqref{Eqn:relative_error} is optimistic,
since it assumes exact $\{d_{t,k,j}^{(I)}\}$ and exact 
$\{\lambda_{j+1}\}$. In practice, the coefficients
$\{d_{t,k,j}^{(I)}\}$ are themselves generated recursively, while the high-order spectral moments $\{\lambda_{j+1}\}$ may also accumulate rounding errors. These additional errors are amplified by the same cancellation mechanism. If catastrophic cancellation occurs, the computed orthogonalization parameters may be unreliable, causing the divergence of the algorithm.

\subsection{Reformulated WS-CG-VAMP}\label{SSec:WS}
The conjugate gradient (CG) method iteratively solves Hermitian positive-definite linear systems by generating search directions that are mutually conjugate with respect to the system matrix. Specializing the general formulation in Section~\ref{Sec:GM_MAMP} to CG recovers the WS-CG-VAMP algorithm proposed in~\cite{skuratovs2022compressed}, thereby
showing that WS-CG-VAMP falls within the MAMP framework.

\textbf{\emph{WS-CG-VAMP}}\cite{skuratovs2022compressed}: Starting with $t = 1$ and $\bm{x}_1^{\phi} = \bm{0}_N$,
\BS\begin{alignat}{2}
    {\rm MLE:} \quad && \quad\bm{x}_t^{\gamma} &= \tfrac{1}{\varepsilon^{\gamma}_{t}}\big(\bm{A}^{\mr H}\bm{u}^{(I)}_{t} + \textstyle\sum_{k=1}^t \omega^{(I)}_{t,k}\bm{x}_k^{\phi} \big), \label{Eqn:WS_MLE} \\
    {\rm NLE:} \quad && \bm{x}_{t+1}^{\phi} &= \phi_t(\bm{x}_t^{\gamma}). \label{Eqn:WS_NLE}
\end{alignat}\ES
Setting $\eta_t^{(i)}=\theta_t^{(i)}= 0$ in \eqref{Eqn:gen_grad} gives the following CG inner iterations: for $0 \leq i < I$,
\BS\begin{align}
    \bm{p}_t^{(i+1)} &= \bm{r}_t^{(i)} + \beta_t^{(i)} \bm{p}_t^{(i)}, \\
    \bm{u}_t^{(i+1)} &= \bm{u}_t^{(i)} + \alpha_t^{(i)}\bm{p}_t^{(i+1)}, \\
    \bm{r}_t^{(i+1)} &= \bm{r}_t^{(i)} - \alpha_t^{(i)}\bm{W}_t\bm{p}_t^{(i+1)}, \label{Eqn:mvcg_0}
\end{align}
where 
\begin{align}
    \alpha_{t}^{(i)} &= \frac{\|\bm{r}_t^{(i)}\|^2} {\big\langle\bm{p}_t^{(i+1)},\bm{W}_t\bm{p}_t^{(i+1)}\big\rangle}, \label{Eqn:mvcg_1} \\
    \beta_t^{(i)} &= \frac{\|\bm{r}_t^{(i)}\|^2}{\|\bm{r}_{t}^{(i-1)}\|^2},\ \; i \geq 1.
\end{align}\ES
The initialization is
\begin{align}
    \bm{u}_t^{(0)} &= \bm{u}_{t-1}^{(I)},\ \,\bm{p}_t^{(0)} = \bm{p}_{t-1}^{(I)},\ \;
    \beta_t^{(0)} = \beta_{t-1}^{(I)},\ \; t \geq 2 \\
    \bm{r}_t^{(0)} &= \bm{z}_t - \bm{W}_t \bm{u}_t^{(0)}, \label{Eqn:mvcg_2}
\end{align}
with $\bm{u}_1^{(0)} = \bm{p}_{1}^{(0)} = \bm{0}_M$ and $\beta_1^{(0)} = 0$. The orthogonalization parameters $\{\omega^{(I)}_{t,k}\}$ can be obtained by setting $\eta_t^{(i)}=\theta_t^{(i)}= 0$ in Lemma~\ref{Lem:c_and_d}.

A direct implementation of WS-CG-VAMP requires $2I+4$ matrix-vector products per iteration, outlined as follows:
\begin{itemize}
    \item Computing $\bm{z}_t = \bm{y} - \bm{A}\bm{x}_t^{\phi}$ requires one.
    \item Computing $\bm{W}_t\bm{u}_{t}^{(0)} = \sigma^2\bm{u}_{t}^{(0)} + v_{t, t}^{\phi} \bm{A}\bm{A}^{\rm H}\bm{u}_{t}^{(0)}$ in \eqref{Eqn:mvcg_2} requires two.
    \item Computing $\bm{W}_t\bm{p}_t^{(i+1)} = \sigma^2\bm{p}_t^{(i+1)} + v_{t, t}^{\phi} \bm{A}\bm{A}^{\rm H}\bm{p}_t^{(i+1)}$ in \eqref{Eqn:mvcg_0}-\eqref{Eqn:mvcg_1} for $0 \leq i < I$ requires $2I$.
    \item Computing $\bm{A}^{\mr H}\bm{u}^{(I)}_{t}$ in \eqref{Eqn:WS_MLE} requires one.
\end{itemize}
Indeed, the computations of $\bm{A}^{\mr H}\bm{u}^{(I)}_{t}$ and $\bm{W}_t\bm{u}_{t}^{(0)}$ are redundant. To eliminate these, we introduce some intermediate variables. Let $\bm{b}_t^{(i)} = \bm{A}^{\rm H} \bm{p}_t^{(i)}$, $\bm{e}_t^{(i)} = \bm{A}\bm{A}^{\rm H} \bm{p}_t^{(i)}$, $\bm{h}_t^{(i)} =  \bm{A}\bm{A}^{\rm H} \bm{u}_t^{(i)}$ and $\bm{q}_t^{(I)} = \bm{A}^{\mr H}\bm{u}^{(I)}_{t}$. This leads to the following equivalent reformulation of WS-CG-VAMP.

\textbf{\emph{Reformulated WS-CG-VAMP}}: Starting with $t = 1$, $\bm{q}_{0}^{(I)} = \bm{0}_N$ and $\bm{x}_1^{\phi} = \bm{0}_N$,
\BS\begin{alignat}{2}
    {\rm MLE:} 
    \quad && \bm{q}_t^{(I)} 
    &= \bm{q}_{t-1}^{(I)} + \textstyle\sum_{i=0}^{I-1}\alpha_t^{(i)}\bm{b}_t^{(i+1)}, \\
    && \quad \bm{x}_t^{\gamma} &= \tfrac{1}{{\varepsilon}^{\gamma}_{t}}\big(\bm{q}_t^{(I)} + \textstyle\sum_{k=1}^t \omega^{(I)}_{t,k}\bm{x}_k^{\phi} \big),  \\
    {\rm NLE:} \quad && \bm{x}_{t+1}^{\phi} &= \phi_t(\bm{x}_t^{\gamma}).
\end{alignat}\ES
where $\bm{b}_t^{(i)}$ and $\alpha_{t}^{(i)}$ are given by the following iterations: for $0\leq i<I$,
\BS\begin{align}
    \bm{p}_t^{(i+1)} &= \bm{r}_t^{(i)} + \beta_t^{(i)}\bm{p}_t^{(i)},  \\
    \bm{b}_t^{(i+1)} &= \bm{A}^{\mr H}\bm{p}_t^{(i+1)}, \\
    \bm{e}_t^{(i+1)} &= \bm{A}\bm{b}_t^{(i+1)}, \\
    \bm{u}_t^{(i+1)} &= \bm{u}_t^{(i)} + \alpha_t^{(i)}\bm{p}_t^{(i+1)},  \\
    \bm{h}_t^{(i+1)} &= \bm{h}_t^{(i)} + \alpha_t^{(i)}\bm{e}_t^{(i+1)},  \\
    \bm{r}_t^{(i+1)} &= \bm{r}_t^{(i)} - \alpha_t^{(i)}
    \Big( \sigma^2\bm{p}_t^{(i+1)} + v_{t,t}^{\phi}\bm{e}_t^{(i+1)} \Big), 
\end{align}
where
\begin{align}
    \alpha_t^{(i)} &=
    \frac{\|\bm{r}_t^{(i)}\|^2}{\sigma^2\|\bm{p}_t^{(i+1)}\|^2 + v_{t,t}^{\phi}\|\bm{b}_t^{(i+1)}\|^2}, \\
    \beta_t^{(i)} &= \frac{\|\bm{r}_t^{(i)}\|^2}{\|\bm{r}_t^{(i-1)}\|^2},\ \; i \geq 1
\end{align}\ES
The initialization is
\begin{align}
    \bm{u}_t^{(0)} &= \bm{u}_{t-1}^{(I)},\ \;\bm{h}_t^{(0)} = \bm{h}_{t-1}^{(I)},\ \; t \geq 2 \\
    \bm{p}_t^{(0)} &= \bm{p}_{t-1}^{(I)},\ \;\beta_t^{(0)} = \beta_{t-1}^{(I)},\ \; t \geq 2 \\
    \bm{r}_t^{(0)} &= \bm{z}_t - \sigma^2\bm{u}_t^{(0)} - v_{t,t}^{\phi}\bm{h}_t^{(0)},
\end{align}
with $\bm{u}_1^{(0)} = \bm{h}_1^{(0)} = \bm{p}_1^{(0)} = \bm{0}_M$ and $\beta_1^{(0)}=0$.

This reformulated algorithm, referred to as WS-CG-VAMP(r),
requires only $2I+1$ matrix-vector products per iteration. For a fair comparison, we evaluate the convergence of GD-MAMP and WS-CG-VAMP(r) in terms of the
number of matrix-vector products rather than the number of iterations. As will be shown in Section~\ref{Sec:sim} (Figures~\ref{Fig:CG_kap}-\ref{Fig:CG_vpa}):
\begin{itemize}
    \item GD-MAMP converges faster than WS-CG-VAMP(r) when the condition number $\kappa(\bm{A})$ is small.
    \item For large $\kappa(\bm A)$,  WS-CG-VAMP(r) converges faster than GD-MAMP under high-precision
    arithmetic. Under IEEE double-precision arithmetic, WS-CG-VAMP(r) may diverge due to catastrophic cancellation in the computation of the orthogonalization parameters, as analyzed in Section~\ref{Sec:gen_precision}. In other words, realizing the convergence advantage of WS-CG-VAMP(r) requires high-precision arithmetic, which incurs substantial runtime overhead.
\end{itemize}
\begin{table}[htb]
\renewcommand{\arraystretch}{1.5} 
\centering \footnotesize \setlength{\tabcolsep}{1mm}{
\caption{Comparison for GD-MAMP and WS-CG-VAMP}
\label{table:CG}
\begin{tabular}{|c|c|c|c|}
\hline
Algorithm & \makebox[2.4cm]{Version} & \makecell{Number of matrix-vector \\ products per iteration} & \makecell{Numerical problem} \\
\hline
\hline
\multirow{2}{*}{GD-MAMP} & original \cite{liu2022memory} & $4$ & overflow (with large $\kappa$) \\
\cline{2-4}
& OA-GD-MAMP (this paper) & $3$ & no \\
\hline
\multirow{2}{*}{WS-CG-VAMP} & original \cite{skuratovs2022compressed} & $2I+4$ & \multirow{2}{*}{\makecell{catastrophic cancellation \\ (with large $\kappa$)}} \\
\cline{2-3}
& WS-CG-VAMP(r) (this paper) & $2I+1$ & \\
\hline
\end{tabular}} 
\end{table}
Table~\ref{table:CG} summarizes the comparison for GD-MAMP and WS-CG-VAMP.

\section{Simulation Results}\label{Sec:sim}
In this section, we present numerical simulations to support our results in Sections \ref{Sec:OA-GD}-\ref{Sec:CG}.

\subsection{Experiment Setup}
We set the matrix $\bm{A}$ and the signal $\bm{x}$ as follows:
\begin{itemize}
    % \item The signal $\bm{x} = [x_1, \cdots\!, x_N]^{\rm T}$ is IID. For the real-valued setting, each element $x_k$ follows a Bernoulli-Gaussian distribution, i.e., for $k \in [N]$, 
    % \begin{align}
    %     x_k = 
    %     \begin{cases}
    %     0, & \text{probability } 1 - \mu \\[1mm]
    %     {\mathcal{N}}(0, {\mu}^{-1}), & \text{probability } \mu,
    %     \end{cases}
    % \end{align}
    % where $\mu \in (0, 1)$ is the sparsity rate. For the complex-valued setting, each element $x_k$ follows a complex Bernoulli-Gaussian distribution, i.e., for $k \in [N]$, 
    % \begin{align}
    %     x_k = 
    %     \begin{cases}
    %     0, & \text{probability } 1 - \mu \\[1mm]
    %     {\mathcal{CN}}(0, {\mu}^{-1}), & \text{probability } \mu.
    %     \end{cases}
    % \end{align}
    % Hence, we have $\mathrm{E}\{\bm{x}\}=\bm{0}$ and $\tfrac{1}{N}{\rm E}\{\|\bm{x}\|^2\}=1$. The signal-to-noise ratio (SNR) is defined as ${\rm SNR} = 1 / \sigma^2$.
    \item The signal $\bm{x}=[x_1,\ldots,x_N]^{\rm T}$ has IID Bernoulli-Gaussian entries: for $k \in [N]$,
    \begin{align}
        x_k = b_k g_k,\quad
        b_k\sim\operatorname{Bern}(\mu),\quad
        g_k\sim
        \begin{cases}
            \mathcal{N}(0,\mu^{-1}), & \text{real-valued setting},\\
            \mathcal{CN}(0,\mu^{-1}), & \text{complex-valued setting},
        \end{cases}
    \end{align}
    where $b_k$ is independent
    of $g_k$. Here, $\mu\in(0,1)$ is the sparsity rate. Hence, $\mathbb{E}[\bm{x}]=\bm{0}$ and $N^{-1}\mathbb{E}[\|\bm{x}\|^2]=1$, and the SNR is
    $\mathrm{SNR}=1/\sigma^2$.
    \item Unless otherwise stated, we use the following sensing matrix $\bm{A}$:
    \begin{align}
        \bm{A} = \bm{\Sigma}\bm{\Pi}\bm{F}\bm{S}, \label{Eqn:spA}
    \end{align}
    where $\bm{\Pi}$ is a uniformly random permutation matrix, $\bm{F}$ is the $N$-point normalized discrete cosine transform (DCT) matrix for the real-valued setting or the discrete Fourier transform (DFT) matrix for the complex-valued setting, $\bm{S}$ is a diagonal matrix with IID Rademacher entries independent of $\bm{\Pi}$, $\bm{\Sigma} \in \mathbb{R}^{M \times N}$ is a rectangular diagonal matrix with diagonal entries $\sigma_1, \cdots, \sigma_J$, where $J = \min(M, N)$. We set $\sigma_i/\sigma_{i+1} = \kappa^{1/J}$ and $\tfrac{1}{N}\textstyle\sum_{i=1}^{J}\sigma_i^2 = 1$. As a result, the condition number of $\bm{A}$ is $\kappa^{1-1/J}$, which tends to $\kappa$ when $J$ is large. %The use of DCT/DFT matrix reduces the simulation time, since performing a matrix-vector product only requires complexity of $\mathcal{O}(N\log N)$ instead of $\mathcal{O}(MN)$.
    Although $\bm{A}$ is not right-unitarily invariant, it belongs to the broader spectral universality class \cite{Rishabh2024, wang2024universality}, under which OAMP/VAMP retains replica Bayes-optimality.
    \item When comparing the runtime of algorithms, in order to reflect the time cost of matrix-vector products, we do not use $\bm{A}$ in \eqref{Eqn:spA}. Instead, we generate a right-unitarily invariant matrix $\bm{A}$ with the same singular values $\sigma_1,\cdots, \sigma_J$ as those in \eqref{Eqn:spA}. 
\end{itemize}
In general, AMP-type algorithms converge more slowly when dealing with more challenging matrices $\bm{A}$. The ``difficulty'' of $\bm{A}$ is determined not only by its condition number but also by the specific distribution of its singular values. Even with the same condition numbers, the matrix $\bm{A}$ with singular values described above may be much more challenging than those encountered in practical applications.

\subsection{Simulation Results}
The simulations in Figs. \ref{Fig:OA}-\ref{Fig:CR_GD_1} are conducted under the complex-valued setting. Recall that the intermediate variables $|w_{2t-2}|$ and $|w_{2t-1}|$ are required in iteration $t$ of GD-MAMP. In Fig.~ \ref{Fig:OA}(a), we show that $|w_k|$ increases exponentially with $k$. When $k \geq 273$, $|w_k|$ exceeds the maximum representable value of standard 64-bit double-precision floating-point (FP64) numbers (approximately $1.8 \times 10^{308}$), leading to the overflow problem. This explains the crash of GD-MAMP at $t=137$ in Fig. \ref{Fig:OA}(b). In contrast, the proposed OA-GD-MAMP in Section \ref{Sec:OA-GD} avoids this overflow problem. When the eigenvalues of $\bm{A}\bm{A}^{\rm H}$ are available, OA-GD-MAMP is mathematically equivalent to GD-MAMP. Furthermore, when the eigenvalues of $\bm{A}\bm{A}^{\rm H}$ are unavailable, OA-GD-MAMP performs well by using the approximation method in Lemma \ref{Lemma:oa_h}.

In the following, unless otherwise specified, we assume that the eigenvalues of $\bm{A}\bm{A}^{\rm H}$ are known. Thus, we drop the distinction with OA-GD-MAMP and write GD-MAMP only. Furthermore, the overflow-avoiding technique is applied to CR-GD-MAMP and its partial memory variant (abbreviated as GD-MAMP-P) as described in Section \ref{Sec:CR_GD}.
\begin{figure}[t] 
  \begin{tabular}{c c}
  \includegraphics[width=75mm]{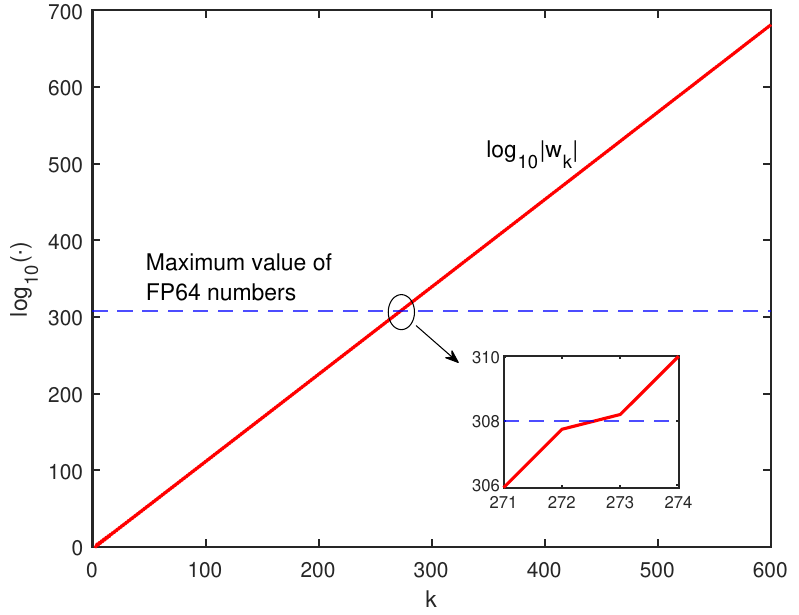} & \includegraphics[width=75mm]{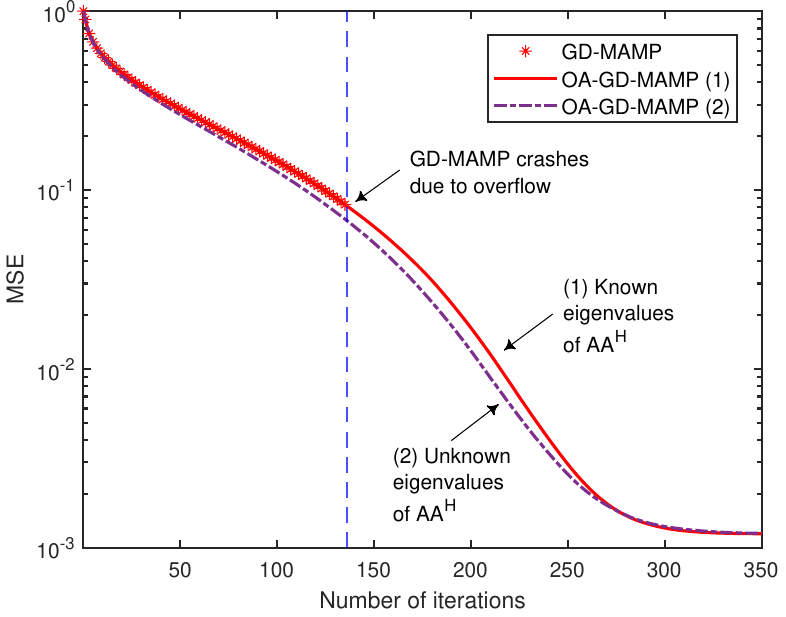} \\
  $\quad$\footnotesize{(a) $\log_{10}|w_k|$ versus $k$} &
  $\quad$\footnotesize{(b) MSE versus number of iterations} \\
  \end{tabular}
  \caption{In (a), $|w_k|$ exceeds the maximum value of FP64 numbers when $k$ is large, leading to the overflow of GD-MAMP in (b). In contrast, OA-GD-MAMP avoids this overflow problem and works well. $M=8192, N=16384$, $\mu=0.1$, ${\rm SNR}=35$ dB, $\kappa = 1000$, $L=3$.}
  \label{Fig:OA}
\end{figure}
\begin{figure}[t] 
  \begin{tabular}{c c}
  \includegraphics[width=75mm]{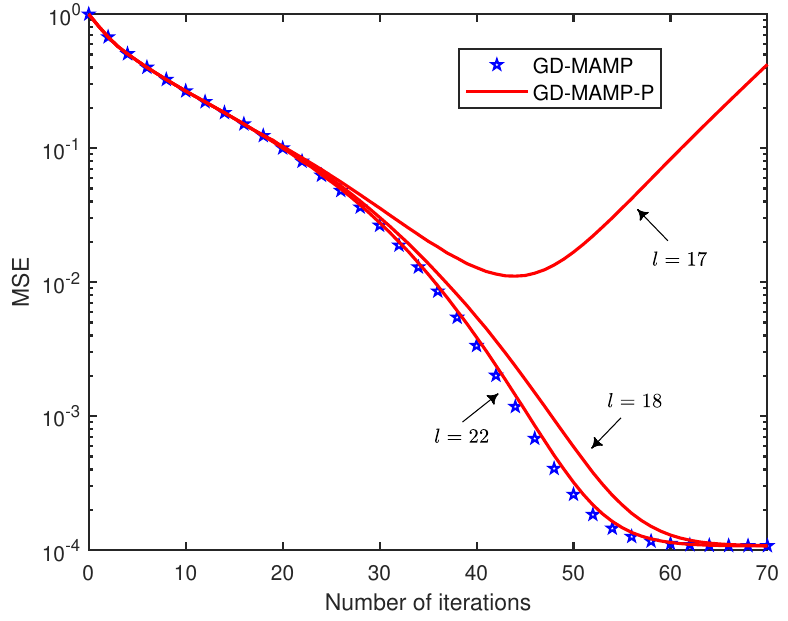} & \includegraphics[width=75mm]{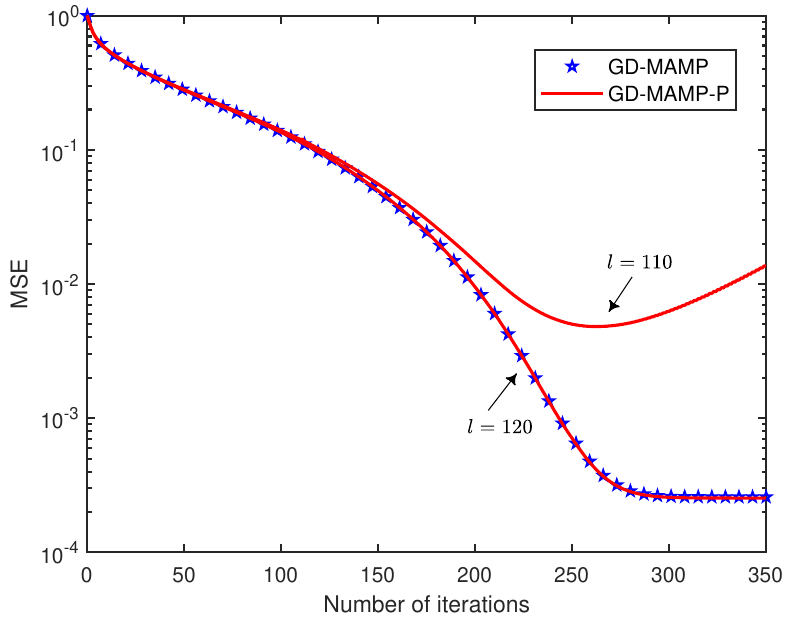} \\
  $\quad$\footnotesize{(a) $\kappa = 100$} &
  $\quad$\footnotesize{(b) $\kappa = 1000$} \\
  \end{tabular}
  \caption{If $\ell$ is suitable (e.g. $\ell = 22$ in (a) and $\ell=120$ in (b)), GD-MAMP-P saves the storage of $(T-\ell)N$ floating numbers while keeping the same convergence as GD-MAMP. $M=8192, N=16384$, $\mu=0.1$, ${\rm SNR}=35$ dB, $L=3$.}
  \label{Fig:part}
\end{figure}

Recall that $\ell$ denotes the maximum number of memory terms in GD-MAMP-P. As shown in Fig. \ref{Fig:part}(a), when $\kappa = 100$, GD-MAMP requires about $T=60$ iterations to converge, and GD-MAMP-P has the similar convergence with $\ell = 22$. In this case, GD-MAMP-P only needs to store $\ell N = 22 \cdot 16384$ floating numbers,  instead of $TN = 60 \cdot 16384$ in GD-MAMP. However, when $\ell < 18$, GD-MAMP-P diverges. In Fig. \ref{Fig:part}(b), when $\kappa = 1000$, GD-MAMP requires about $T=300$ iterations to converge, while GD-MAMP-P has the same convergence with $\ell = 120$. Based on these observations, it might be safe to set $\ell$ to $40\%$-$50\%$ of the maximum number of iterations of GD-MAMP, giving a $50\%$-$60\%$ decrease in storage of memory terms.
\begin{figure}[t] 
  \begin{tabular}{c c}
  \includegraphics[width=75mm]{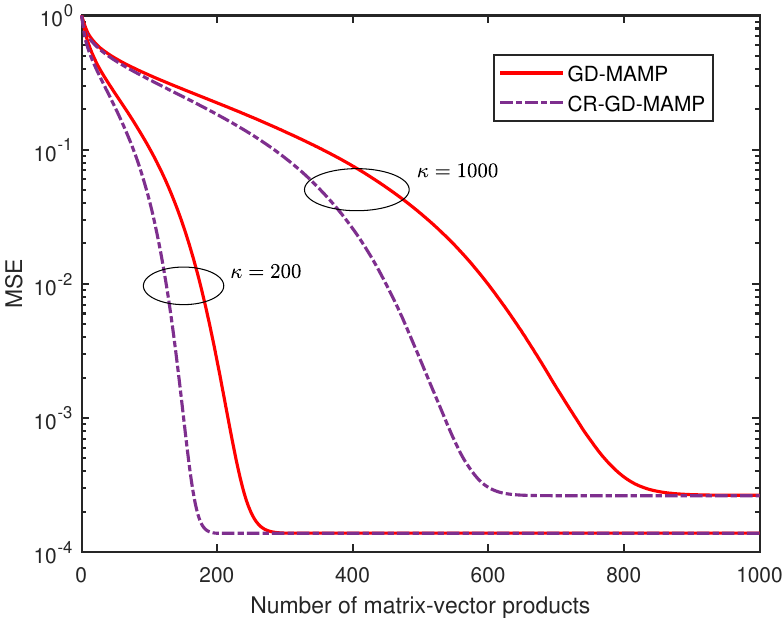} & \includegraphics[width=75mm]{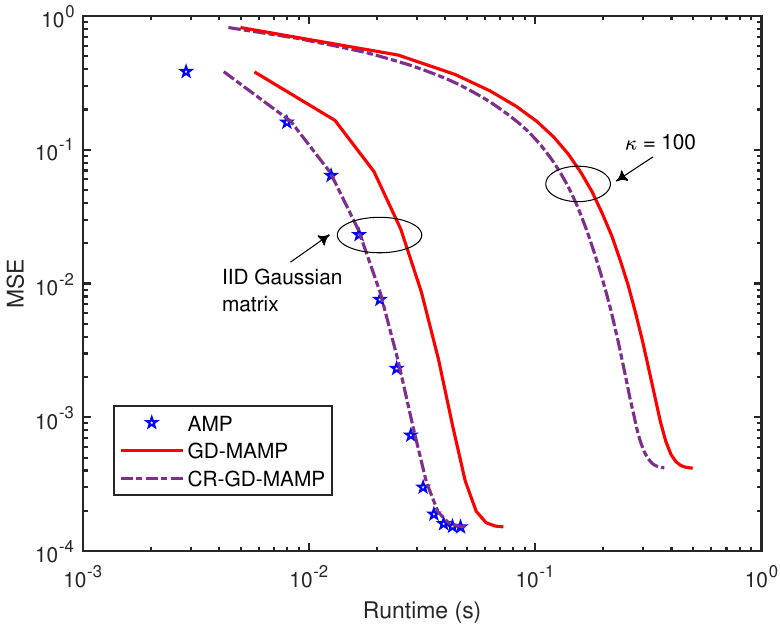} \\
  $\quad$\footnotesize{(a) MSE versus number of matrix-vector products,} & $\quad$\footnotesize{(b) MSE versus runtime,} \\
  $\quad$\footnotesize{$\kappa \in \{200, 1000\}$, $N=16384$} 
  & $\quad$\footnotesize{$\kappa \in \{15, 100\}$, $N=4096$} \\
  \end{tabular}
  \caption{In (a), CR-GD-MAMP saves nearly $1/3$ matrix-vector products while keeping the similar convergence as GD-MAMP. In (b), the runtime of CR-GD-MAMP is the same as that of AMP (for IID Gaussian $\bm{A}$), and only $2/3$ the runtime of GD-MAMP. $M=0.5N$, $\mu=0.1$, ${\rm SNR}=35$ dB, $L=3$.}
  \label{Fig:CR_GD_1}
\end{figure}
\begin{figure}[t] 
    \centering
    \includegraphics[width=80mm]{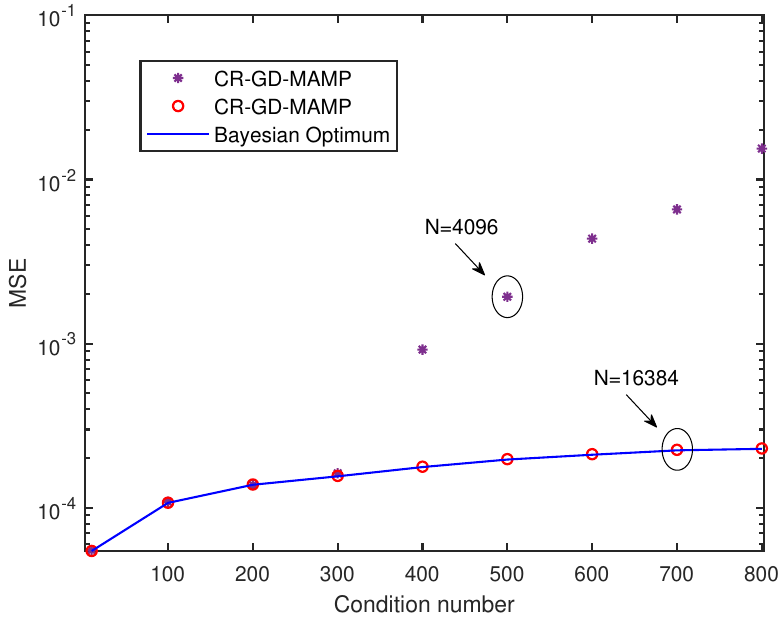} \\
    \caption{When $\kappa$ is large and $N$ is not large enough, CR-GD-MAMP may not converge to the same performance as GD-MAMP. $N \in \{4096, 16384\}$, $M = 0.5N$, $\mu=0.1$, ${\rm SNR}=35$ dB, $L=3$.}
    \label{Fig:CR_GD_2}
\end{figure}

Next, we compare the convergence speed of GD-MAMP and CR-GD-MAMP by evaluating the number of matrix-vector products. As shown in Fig. \ref{Fig:CR_GD_1}(a), CR-GD-MAMP converges to the same point as GD-MAMP with about $2/3$ the number of matrix-vector products, supporting our analysis in Section \ref{SSec:CR}. Fig. \ref{Fig:CR_GD_1}(b) shows that the average runtime of CR-GD-MAMP is similar to AMP (for IID Gaussian matrices $\bm{A}$), and only $2/3$ the runtime of GD-MAMP. This result is consistent with our expectation that matrix-vector products dominate the time complexity. In addition, Fig. \ref{Fig:CR_GD_2} shows that CR-GD-MAMP fails to converge to Bayesian optimum for a large $\kappa$ when the system scale decreases (e.g. $N=4096$, $\kappa > 300$). The reason is the estimation method in Lemma \ref{Lem:CR_cov} is less robust when $N$ becomes small, causing bad damping in CR-GD-MAMP.

\begin{figure}[h] 
  \begin{tabular}{c c}
  \includegraphics[width=75mm]{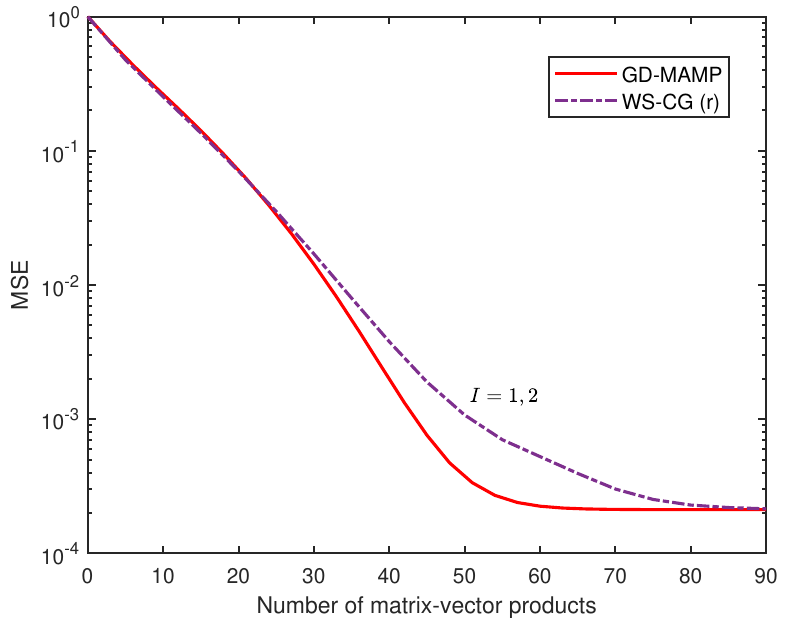} & \includegraphics[width=75mm]{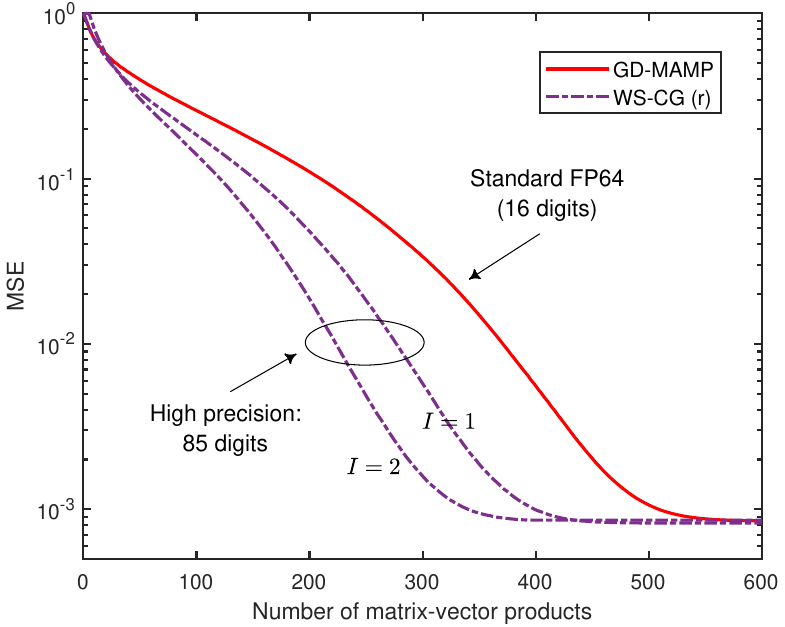} \\
  \footnotesize{(a) $\kappa = 15$} &
  \footnotesize{(b) $\kappa = 500$} \\
  \end{tabular}
  \caption{WS-CG(r) is the short of reformulated WS-CG-VAMP. In (a), GD-MAMP converges faster than WS-CG(r) when $\kappa$ is small. In (b), WS-CG(r) converges faster, but requiring high precision arithmetic, leading to a high cost. $M=8192, N=16384$, $\mu=0.1$, ${\rm SNR}=30$ dB, $L=3$.}
  \label{Fig:CG_kap}
\end{figure}

\begin{figure}[t] 
    \centering
    \includegraphics[width=80mm]{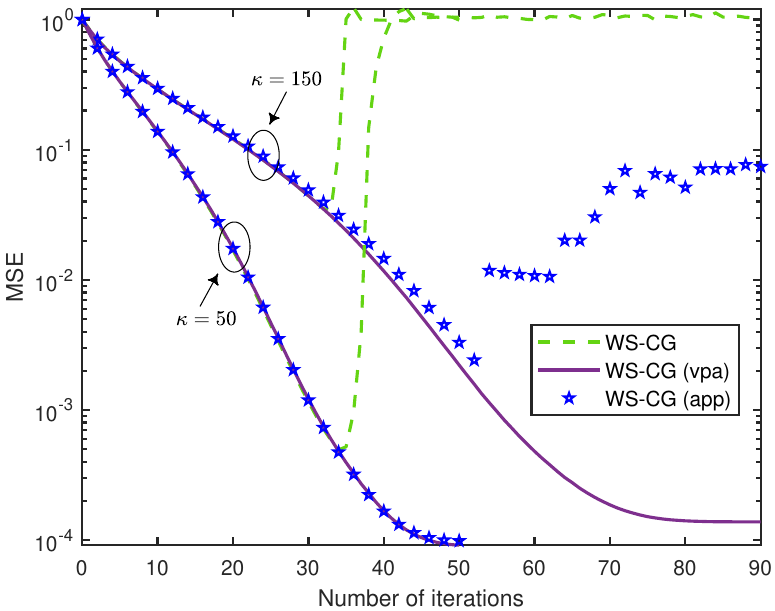}
    \caption{WS-CG(vpa)/(app) denotes WS-CG-VAMP with high-precision arithmetic and the approximation method in \cite{skuratovs2022compressed}, respectively. When $\kappa$ is large, only high-precision arithmetic can fully solve the precision problem of WS-CG-VAMP. $M = 8192$, $N = 16384$, ${\rm SNR}=35$ dB, $\kappa \in \{50, 150\}$, $I=1$.}
    \label{Fig:CG_vpa}
\end{figure}

\begin{figure}[t]
    \centering
    \includegraphics[width=80mm]{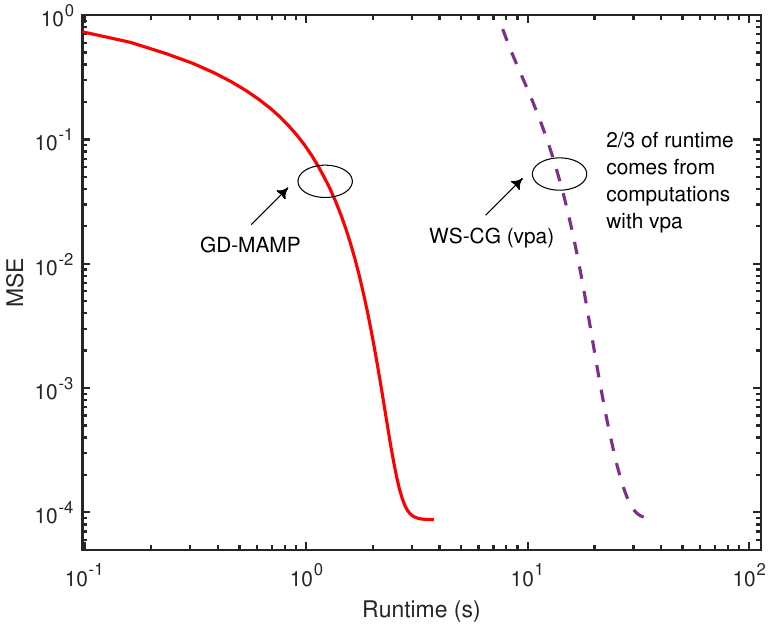}
    \caption{WS-CG(vpa) denotes WS-CG-VAMP with high-precision arithmetic. The practical time cost of WS-CG(vpa) is much higher than GD-MAMP. $M = 8192$, $N = 16384$, ${\rm SNR}=35$ dB, $\kappa = 50$, $L=3$, $I=1$.}
    \label{Fig:CG_time}
\end{figure}

Starting from Fig. \ref{Fig:CG_kap}, the simulations are conducted under the real-valued setting. We compare GD-MAMP and reformulated WS-CG-VAMP (abbreviated as WS-CG(r)) in Section \ref{Sec:CG}, by evaluating the number of matrix-vector products. As shown in Fig. \ref{Fig:CG_kap}(a), GD-MAMP converges faster than WS-CG(r) with $I=1$ or $I=2$ when $\kappa=15$. On the contrary, as shown in Fig. \ref{Fig:CG_kap}(b), WS-CG(r) converges faster than GD-MAMP when $\kappa=500$. However, in this case, high precision arithmetic for WS-CG(r) is required.

Fig.~\ref{Fig:CG_vpa} shows the necessity of using high-precision arithmetic for WS-CG-VAMP under moderate to large $\kappa$. Let WS-CG(vpa)/(app) be WS-CG-VAMP with high-precision arithmetic and the approximation method in \cite{skuratovs2022compressed}, respectively. When $\kappa=50$, WS-CG diverges due to the precision problem when calculating the orthogonalization parameters. In this case, both WS-CG(vpa) and  WS-CG(app) converge. When $\kappa = 150$, WS-CG(app) is not robust. Only with high-precision arithmetic can WS-CG-VAMP overcome its precision problem and converge. However, the time cost of using high-precision arithmetic is quite high, which is typically not accepted in practice. As shown in Fig.~\ref{Fig:CG_time}, the runtime of WS-CG(vpa) is much higher than that of GD-MAMP, even though their theoretical complexity is similar.

\section{Conclusion}
%In this paper, we first addressed the potential overflow problem in the original GD-MAMP, ensuring its numerical stability for systems with large condition numbers. Second, we introduced two complexity-reduced variants of GD-MAMP, achieving less memory terms and $2/3$ the number of matrix-vector products while keeping the similar convergence, respectively. Next, we proposed a general gradient-based formulation for designing MAMP algorithms and show that WS-CG-VAMP can be unified under this framework. Furthermore, we provide an equivalent reformulation of WS-CG-VAMP to reduce the number of matrix-vector products, so that we can conduct a comprehensive comparison with GD-MAMP. For small condition numbers, GD-MAMP achieves superior convergence speed. Conversely, for large condition numbers, although GD-MAMP exhibits a slower convergence rate, it does not suffer from the numerical precision problem in WS-CG-VAMP.

In this paper, we first resolved the overflow problem in the
original GD-MAMP by deriving an equivalent overflow-avoiding
version, thereby preventing numerical overflow for systems with large condition numbers. Second, we developed two
low-complexity variants: GD-MAMP with partial memory terms, which reduces memory requirements, and CR-GD-MAMP, which reduces the number of matrix-vector products per outer iteration from three to two. Both variants retain the similar convergence speed as GD-MAMP.

Next, we developed a general gradient-based formulation for
designing MAMP algorithms and showed that WS-CG-VAMP can be
recovered as its CG specialization. We identified catastrophic cancellation in the computation of the orthogonalization parameters as the mechanism underlying the finite-precision instability of WS-CG-VAMP, and quantified the resulting roundoff-error amplification using a cancellation factor. Furthermore, we derived an equivalent reformulation of WS-CG-VAMP, termed WS-CG-VAMP(r), which reduces the number of matrix-vector products per outer iteration from $2I+4$ to $2I+1$, where $I \geq 1$ is the number of inner CG iterations. In terms of the number of matrix-vector products, GD-MAMP converges faster for small condition numbers. For large condition numbers, WS-CG-VAMP(r) converges faster under high-precision arithmetic. Under IEEE double-precision arithmetic, however, WS-CG-VAMP(r) may diverge because of catastrophic cancellation, whereas GD-MAMP does not exhibit this instability.

Developing numerically stable IEEE double-precision
implementations of MAMP algorithms derived from the proposed
gradient-based formulation, including WS-CG-VAMP, remains an
important direction for future work.

%\section*{Acknowledgment}  

\appendices

\section{Details and Pseudocode of OA-GD-MAMP}\label{App:OA}
\subsection{Details of OA-GD-MAMP}
For simplicity, we define 
\begin{align}
    \vartheta_{t,i}^{\log} &\equiv \log|\vartheta_{t,i}|,
\end{align}
for $1 \leq i \leq t$. The following expressions in GD-MAMP need to be rewritten by \eqref{Eqn:bw}:
\begin{itemize}
    \item 
    Note that $w_0 = \rchi_0$ and $w_1 = \rchi_1/\theta_0$. Thus, $c_{t,0}, c_{t,1}, c_{t,2}, c_{t,3}$ in \eqref{Eqn:ct0_4} are rewritten as
    \BS\label{Eqn:oa_ct}\begin{align}
        c_{t,0} &= \textstyle\sum_{i=1}^{t-1}{\rm sgn}(\xi_i) \mathrm{e}^{\vartheta_{t,i}^{\log}-(t-i)\log\theta_0} \rchi_{t-i} / \rchi_0,\\
        c_{t,1} &= \sigma^2\rchi_0+ v_{t,t}^{\bar{\phi}}\big(\lambda^\dag\rchi_0 - \rchi_1/\theta_0 - \rchi_0^2\big), \\
        c_{t,2} &= \textstyle\sum_{i=1}^{t-1} {\rm sgn}(\xi_i)\mathrm{e}^{\vartheta_{t,i}^{\log}-(t-i)\log\theta_0} \big[\Re(v_{t,i}^{\bar{\phi}})\rchi_{t-i+1}/\theta_0-\big(\sigma^2+\Re(v_{t,i}^{\bar{\phi}})(\lambda^\dag-\rchi_0)\big)\rchi_{t-i}\big]\\
        \begin{split}
            c_{t,3} &= \textstyle{\sum}_{i=1}^{t-1}\textstyle{\sum}_{j=1}^{t-1} {\rm sgn}(\xi_i\xi_j) \mathrm{e}^{\vartheta_{t,i}^{\log}+\vartheta_{t,j}^{\log}-(2t-i-j)\log\theta_0} \big[\big(\sigma^2+v_{i,j}^{\bar{\phi}}\lambda^\dag\big)\rchi_{2t-i-j} \\
            & \qquad \qquad \qquad -v_{i,j}^{\bar{\phi}}\big(\rchi_{2t-i-j+1}/\theta_0+\rchi_{t-i}\rchi_{t-j}\big) \big]
        \end{split}
    \end{align}\ES
    \item 
    $p_{t,i}$ in \eqref{Eqn:p_ti} is rewritten as
    \begin{align}
        p_{t,i} = {\rm sgn}(\xi_i) \mathrm{e}^{\vartheta_{t,i}^{\log}-(t-i)\log\theta_0} \rchi_{t-i}. \label{Eqn:oa_pti}
    \end{align}
    \item 
    $v_{t,t'}^{\gamma}$ in \eqref{Eqn:v_gam} is rewritten as
    \begin{align}
        \begin{split}
            v_{t,t'}^{\gamma} &= \tfrac{1}{\varepsilon_t^{\gamma}}\tfrac{1}{\varepsilon_{t'}^{\gamma}} \textstyle{\sum}_{i=1}^{t}\textstyle{\sum}_{j=1}^{t'} {\rm sgn}(\xi_i\xi_j) \mathrm{e}^{\vartheta_{t,i}^{\log}+\vartheta_{t',j}^{\log}-(t+t'-i-j)\log\theta_0}\big[\big(\sigma^2+v_{i,j}^{\bar{\phi}}\lambda^\dag\big)\rchi_{t+t'-i-j} \\
            & \qquad \qquad \qquad -v_{i,j}^{\bar{\phi}}\big(\rchi_{t+t'-i-j+1}/\theta_0+\rchi_{t-i}\rchi_{t'-j}\big) \big]
        \end{split}
    \end{align}
\end{itemize}

\subsection{Pseudocode of OA-GD-MAMP}
\begin{algorithm}[H]
    \caption{OA-GD-MAMP}\label{alg:cap}
    \begin{algorithmic}[1]
        \renewcommand{\algorithmicrequire}{\textbf{Input:}}
        \renewcommand{\algorithmicensure}{\textbf{Output:}}
        \Require $\bm{A}$, $\bm{y}$, $\sigma^2$, $M$, $N$, $L$, $T$ 
        %\vspace{1mm}
        \If{$\bm{\lambda}$ (eigenvalues of $\bm{A}\bm{A}^{\rm H}$) is known}
            %\vspace{1mm}
            \State $\lambda^{\dag} = (\lambda_{\max}+\lambda_{\min}) / 2$, \  $\bm{\lambda}_{B} = \lambda^{\dag}\bm{1}-\bm{\lambda}$, \ $\rchi_0 = \frac{1}{N}\bm{1}^{\rm T}\bm{\lambda}$
            %\vspace{1mm}
            \State $\big\{\rchi_k, 1 \leq k < 2T\big\}$ by Lemma \ref{Lemma:oa_1}
            %\vspace{1mm}
        \Else
            %\vspace{1mm}
            \State $\lambda_{\max}^{\rm up}$ and $\lambda_{\min}^{\rm low}$ by (\ref{Eqn:max_min}), \ $\lambda^{\dag} = (\lambda_{\max}^{\rm up}+\lambda_{\min}^{\rm low}) / 2$
            %\vspace{1mm}
            \State $\big\{\rchi_k, 0 \leq k < 2T\big\}$ by Lemma \ref{Lemma:oa_h}
            %\vspace{1mm}
        \EndIf
        %\vspace{1mm}
        \State $\bm{x}_1 = {\rm E}(\bm{x}) = \bm{0}_N$, \ $\bm{z}_1 = \bm{y} - \bm{A}\bm{x}_1$, \ $\delta = M/N$
        %\vspace{1mm}
        \State $v_{1, 1}^{\bar{\phi}} = (\tfrac{1}{N}\bm{z}_1^{\rm H}\bm{z}_1 - \delta\sigma^2) / \rchi_0$, \ $\xi_1 = 1$, \ $\bm{u}_0 = \bm{0}_M$, \ $\hat{\bm{r}}_0 = \bm{0}_N$.
        %\vspace{1mm}
        \For{$t = 1$ to $T$}
            %\vspace{1mm}
            \State \% MLE
            %\vspace{1mm}
            \State $\theta_t = (\lambda^{\dag} + \sigma^2/v_{t,t}^{\bar{\phi}})^{-1}$, \ \big\{$\vartheta_{t,i}^{\log} = \vartheta_{t-1,i}^{\log} + \log(\theta_t), 1 \leq i < t$\big\}
            %\vspace{1mm}
            \State \big\{$c_{t, 0}, \cdots\!, c_{t, 3}$\big\} by (\ref{Eqn:oa_ct}), \ \big\{$\xi_t = (c_{t,2}c_{t,0} + c_{t,3})/(c_{t,1}c_{t,0} + c_{t,2}), t > 1$\big\}
            %\vspace{1mm}
            \State $\vartheta_{t,t}^{\log} = \log|\xi_t|$, \ $\big\{p_{t,i}, 1 \leq i \leq t\big\}$ by (\ref{Eqn:oa_pti}), \ $\varepsilon_t^{\gamma} = \textstyle\sum_{i=1}^t p_{t,i}$
            %\vspace{1mm}
            \State $\bm{u}_t = \theta_t\lambda^{\dag}\bm{u}_{t-1} + \xi_t \bm{y} - \bm{A}(\theta_t\hat{\bm{r}}_{t-1} + \xi_t\bm{x}_t)$, \ $\hat{\bm{r}}_t = \bm{A}^{\rm H}\bm{u}_t$ 
            %\vspace{1mm}
            \State $\bm{r}_t = \tfrac{1}{{\varepsilon}^\gamma_t}(\hat{\bm{r}}_{t} + \textstyle\sum_{i=1}^t  p_{t, i}\bm{x}_i)$, \ $v_{t,t}^{\gamma} = \tfrac{1}{({\varepsilon}^\gamma_t)^2}(c_{t,1}\xi_t^2 - 2 c_{t,2}\xi_t + c_{t,3})$
            %\vspace{1mm}
            \State \% NLE
            %\vspace{1mm}
            \State $\big\{\hat{\bm{x}}_t, \hat{v}_{t, t}\big\} = \hat{\phi}_t(\bm{r}_t, v_{t,t}^{\gamma})$, \ $\bm{x}_{t+1} = (   \hat{\bm{x}}_t / \hat{v}_{t, t} -  \bm{r}_t / v_{t,t}^{\gamma}) / (1 / \hat{v}_{t, t} - 1 / v_{t,t}^{\gamma})$
            %\vspace{1mm}
            \State $\bm{z}_{t+1} = \bm{y} - \bm{A}\bm{x}_{t+1}$, \ $ \big\{ v^{\bar{\phi}}_{t+1,k} = (v^{\bar{\phi}}_{k,t+1})^* = (\tfrac{1}{N}\bm{z}_{t+1} ^{\mr H} \bm{z}_{k} - \delta \sigma^2) / \rchi_0, 1 \leq k \leq t\!+\!1 \big\}$
            %\vspace{1mm}
            \State \% Damping at NLE
            %\vspace{1mm}
            \State $l_{t+1} = \min(L, t\!+\!1)$, \ $\bm{\mathcal{V}}_{t+1} \gets$ bottom-right $l_{t+1} \times l_{t+1}$ submatrix of $\bm{V}_{t+1}^{\bar{\phi}}$ 
            %\vspace{1mm}
            \State $\scaleto{\bm{\zeta}}{8pt}_{t+1} = \big(\bm{\mathcal{V}}_{t+1}^{-1}\bm{1}\big) \big/ \big(\bm{1}^{\rm T}\bm{\mathcal{V}}_{t+1}^{-1}\bm{1})$, \ $v_{t+1, t+1}^{\bar{\phi}} = 1 \big/ \big(\bm{1}^{\rm T}\bm{\mathcal{V}}_{t+1}^{-1}\bm{1})$
            %\vspace{1mm}
            \State $\bm{x}_{t+1} = \textstyle\sum_{i=1}^{l_{t+1}}{\scaleto{\bm{\zeta}}{8pt}}_{t+1}[i]\bm{x}_{t-l_{t+1}+1+i}$, \ $\bm{z}_{t+1} = \textstyle\sum_{i=1}^{l_{t+1}}{\scaleto{\bm{\zeta}}{8pt}}_{t+1}[i]\bm{z}_{t-l_{t+1}+1+i}$
            %\vspace{1mm}
            \State $\big\{v^{{\bar{\phi}}}_{t+1,k} = (v^{{\bar{\phi}}}_{k,t+1})^* =   \textstyle\sum_{i=1}^{l_{t+1}} {\scaleto{\bm{\zeta}}{8pt}}_{t+1}[i]v^{{\bar{\phi}}}_{t-l_{t+1}+1+i,k}, 1 \leq  k \leq t \big\}$
            %\vspace{1mm}
        \EndFor
        %\vspace{1mm}
        \Ensure $\hat{\bm{x}}_t, \hat{v}_{t, t}$
    \end{algorithmic}
\end{algorithm}
 
\section{Proof of Lemma \ref{Lem:gen_MAMP}}\label{App:gen_MAMP}
For $\tau \in \mathbb{Z}$, let $\mathcal{P}_{\tau}(\bm{A}\bm{A}^{\rm H})$ denote the set of polynomials in $\bm{A}\bm{A}^{\rm H}$ of degree at most $\tau$ with complex-valued coefficients, with the convention that $\mathcal{P}_{\tau}(\bm{A}\bm{A}^{\rm H}) = \{\bm{0}_{M \times M}\}$ for $\tau < 0$. For $t \geq 1$ and $s \in \mathbb{Z}$, define
\begin{align}
    \mathcal{H}_{t, s} \equiv \bigg\{\sum_{k=1}^t \bm{H}_{t, k}\bm{z}_k: \bm{H}_{t, k} \in \mathcal{P}_{(t-k)I+s-1}(\bm{A}\bm{A}^{\rm H}) \bigg\}. \label{Eqn:H_ts}
\end{align}
Each $\mathcal{H}_{t,s}$ is a vector space, and the following properties hold:
\begin{align}
    \mathcal{H}_{t, s} &\subseteq \mathcal{H}_{t, s+1}, \\
    \mathcal{H}_{t, s} &= \mathcal{H}_{t+1, s-I} \label{Eqn:p2_H},\ \; s \leq I.
\end{align} 
Thus, it suffices to prove that
\begin{align}
    \bm{u}_t^{(i)},\; \bm{p}_t^{(i)} \in \mathcal{H}_{t, i}, \label{Eqn:proof_Hti}
\end{align}
for $0 \leq i \leq I$ and $t \geq 1$. We use a nested induction: an outer induction on $t$ and an inner induction on $i$ for each fixed $t$.

First, we establish the inner induction step with respect to $i$. Fix $t$, and suppose that
\begin{align}
    \bm{u}_t^{(i-1)} \in \mathcal{H}_{t, i-1},\ \, \bm{u}_t^{(i)}, \bm{p}_t^{(i)} \in \mathcal{H}_{t, i},
\end{align}
for some $0 \leq i < I$. Since $\mathcal{H}_{t, i-1} \subseteq \mathcal{H}_{t, i}$,
\begin{align}
    \Delta\bm{u}_t^{(i)} &= \bm{u}_t^{(i)} - \bm{u}_t^{(i-1)} \in \mathcal{H}_{t, i}, \\
    \hat{\bm{u}}_t^{(i)} &= \bm{u}_t^{(i)} + \eta_t^{(i)}\Delta\bm{u}_t^{(i)} \in \mathcal{H}_{t, i}.
\end{align}
Since $\bm{W}_t = \sigma^2 \bm{I} + v_{t, t}^{\phi}\bm{A}\bm{A}^{\rm H}$ is a polynomial in $\bm{A}\bm{A}^{\rm H}$ of degree one and $\bm{z}_t \in \mathcal{H}_{t, i+1}$, we have
\begin{align}
    \hat{\bm{r}}_t^{(i)} &= \bm{z}_t - \bm{W}_t\hat{\bm{u}}_t^{(i)} \in \mathcal{H}_{t, i+1}, \\
    \bm{p}_t^{(i+1)} &= \hat{\bm{r}}_t^{(i)} + \beta_t^{(i)} \bm{p}_t^{(i)} \in \mathcal{H}_{t, i+1}, \\
    \bm{u}_{t}^{(i+1)} &= \bm{u}_t^{(i)} + \theta_t^{(i)} \Delta\bm{u}_t^{(i)} + \alpha_t^{(i)}\bm{p}_t^{(i+1)} \in \mathcal{H}_{t, i+1}.
\end{align}

Second, we establish the base case and the induction step with respect to $t$. For $t=1$,
\begin{align}
    \bm{u}_1^{(0)} = \bm{u}_1^{(-1)} = \bm{p}_{1}^{(0)} = \bm{0}_M \in \mathcal{H}_{1, -1} = \mathcal{H}_{1, 0}. \label{Eqn:outer_1}
\end{align}
The inner induction initialized by \eqref{Eqn:outer_1} proves \eqref{Eqn:proof_Hti} for $t=1$. Next, suppose that \eqref{Eqn:proof_Hti} holds for some $t \geq 1$. Then, in particular,
\begin{align}
    \bm{u}_t^{(I-1)} \in \mathcal{H}_{t, I-1},\ \;\bm{u}_t^{(I)}, \bm{p}_t^{(I)} \in \mathcal{H}_{t, I}.
\end{align}
Hence, using \eqref{Eqn:ini} and \eqref{Eqn:p2_H}, we have
\BS\label{Eqn:outer_t}\begin{align}
    \bm{u}_{t+1}^{(-1)} &= \bm{u}_t^{(I-1)} \in \mathcal{H}_{t, I-1} = \mathcal{H}_{t+1, -1}, \\
    \bm{u}_{t+1}^{(0)} &= \bm{u}_t^{(I)} \in \mathcal{H}_{t, I} = \mathcal{H}_{t+1, 0}, \\
    \bm{p}_{t+1}^{(0)} &= \bm{p}_t^{(I)} \in \mathcal{H}_{t, I} = \mathcal{H}_{t+1, 0}.
\end{align}\ES
The inner induction initialized by \eqref{Eqn:outer_t} yields
\begin{align}
    \bm{u}_{t+1}^{(i)},\; \bm{p}_{t+1}^{(i)}  \in \mathcal{H}_{t+1, i}.
\end{align}
Thus, \eqref{Eqn:proof_Hti} holds for $t+1$. This completes the proof.

\section{Derivation of Equation~(\ref{Eqn:gen_var})}\label{App:Var}
Let $\bm{x}_{t}^{\phi} = \bm{x} + \bm{f}_{t}$ and $\bm{x}_{t}^{\gamma} = \bm{x} + \bm{g}_{t}$. We have
\BS\begin{align}
    \bm{g}_{t} &= \tfrac{1}{\varepsilon^{\gamma}_{t}}\big[\bm{A}^{\mr H}\bm{u}^{(I)}_{t} + \textstyle\sum_{k=1}^t \omega^{(I)}_{t,k}\bm{x}_k^{\phi} \big] - \bm{x} \notag \\
    &= \tfrac{1}{\varepsilon^{\gamma}_{t}}\textstyle\sum_{k=1}^t \big[\bm{A}^{\mr H}\bm{D}_{t,k}^{(I)}(\bm{y} - \bm{A}\bm{x}_k^{\phi}) + \omega^{(I)}_{t,k}\bm{x}_k^{\phi} \big] - \bm{x} \label{Eqn:var_a}\\
    &= \tfrac{1}{\varepsilon^{\gamma}_{t}}\textstyle\sum_{k=1}^t \big[\bm{A}^{\mr H}\bm{D}_{t,k}^{(I)}(\bm{y} - \bm{A}\bm{x} - \bm{A}\bm{f}_k) + \omega^{(I)}_{t,k}(\bm{x} + \bm{f}_k) \big] - \bm{x} \\
    &= \tfrac{1}{\varepsilon^{\gamma}_{t}}\textstyle\sum_{k=1}^t \big[\bm{A}^{\mr H}\bm{D}_{t,k}^{(I)}(\bm{n} -\bm{A}\bm{f}_k) + \omega^{(I)}_{t,k}\bm{f}_k\big],\label{Eqn:var_c}
\end{align}\ES
where \eqref{Eqn:var_c} follows from $\varepsilon^{\gamma}_{t} = \sum_{k=1}^t \omega^{(I)}_{t,k}$. 
Let $\bm{Q}_t = \sum_{k=1}^t \bm{A}^{\mr H}\bm{D}_{t,k}^{(I)}$, and $\bm{P}_{t, k} = \omega^{(I)}_{t,k}\bm{I} - \bm{A}^{\mr H}\bm{D}_{t,k}^{(I)}\bm{A}$. For $t' \leq t$,
\BS\begin{align}
    v^{\gamma}_{t,t'} &\equiv \tfrac{1}{N}{\mr E} \{\bm{g}_{t}^{\mr H} \bm{g}_{t'} \} \notag \\
    &= \tfrac{1}{N (\varepsilon_{t}^{\gamma})^*\varepsilon_{t'}^{\gamma}} \mr{E} \big\{(\bm{Q}_{t}\bm{n} + \textstyle\sum_{k=1}^t\bm{P}_{t,k}\bm{f}_k)^{\rm H} (\bm{Q}_{t'}\bm{n} + \textstyle\sum_{k=1}^{t'}\bm{P}_{t',k}\bm{f}_k) \big\} \\
    &\overset{\mr{a.s.}}{=} \tfrac{1}{(\varepsilon_{t}^{\gamma})^*\varepsilon_{t'}^{\gamma}} \big(\tfrac{1}{N}\mr{tr}\{\bm{Q}_{t}^{\mr H}\bm{Q}_{t'}\}\sigma^2 + \textstyle\sum_{k_1=1}^t\sum_{k_2=1}^{t'} \tfrac{1}{N}\mr{tr}\{\bm{P}_{t,k_1}^{\rm H}\bm{P}_{t',k_2}\}v_{k_1,k_2}^{\phi}\big). \label{Eqn:var_d}
\end{align}\ES
Next,
\BS\begin{align}
    \tfrac{1}{N}\mr{tr}\{\bm{Q}_{t}^{\mr H}\bm{Q}_{t'}\} &=\tfrac{1}{N}\textstyle\sum_{k_1=1}^t\sum_{k_2=1}^{t'}\mr{tr}\Big\{\big(\bm{D}_{t,k_1}^{(I)}\big)^{\rm H}\bm{A}\bm{A}^{\rm H}\bm{D}_{t',k_2}^{(I)}\Big\} \\
    &= \textstyle\sum_{k_1=1}^{t}\sum_{k_2=1}^{t'}\sum_{j_1=0}^{J_{t, k_1}}\sum_{j_2=0}^{J_{t', k_2}} \big(d^{(I)}_{t,k_1,j_1}\big)^{*}d_{t',k_2,j_2}^{(I)}\lambda_{j_1+j_2+1}, \label{Eqn:var_e} \\
    \tfrac{1}{N}\mr{tr}\big\{\bm{P}_{t,k_1}^{\rm H}\bm{P}_{t',k_2}\big\} 
    &= \tfrac{1}{N}\mr{tr}\big\{\big(\omega_{t,k_1}^{(I)}\bm{I}-\bm{A}^{\rm H}\bm{D}_{t,k_1}^{(I)}\bm{A})^{\rm H}(\omega_{t',k_2}^{(I)}\bm{I}-\bm{A}^{\rm H}\bm{D}_{t',k_2}^{(I)}\bm{A}\big)\big\} \\
    &= \textstyle\sum_{j_1=0}^{J_{t, k_1}}\sum_{j_2=0}^{J_{t', k_2}}  (d_{t,k_1,j_1}^{(I)})^* d_{t',k_2,j_2}^{(I)}\tilde{\lambda}_{j_1,j_2}, \label{Eqn:var_f}
\end{align}\ES
where $J_{t, k} \equiv (t-k+1)I-1$, $\tilde{\lambda}_{j_1,j_2} \equiv \lambda_{j_1+j_2+2}-\lambda_{j_1+1}\lambda_{j_2+1}$. Thus, we can obtain (\ref{Eqn:gen_var}) by substituting \eqref{Eqn:var_e} and \eqref{Eqn:var_f} into \eqref{Eqn:var_d}.

\section{Proof of Lemma~\ref{Lem:round}}\label{App:proof_round}
For notational simplicity, let
\begin{align}
    J &= J_{t,k}, \\
    \tau_j &\equiv \tau_{t,k,j}^{(I)} = d_{t,k,j}^{(I)}\lambda_{j+1},\ \; 0 \leq j\leq J.
\end{align}
Then, $\omega_{t,k}^{(I)}$ can be written as
\begin{align}
    \omega_{t,k}^{(I)} = \sum_{j=0}^{J}\tau_j.
\end{align}
Under the standard floating-point model,
\begin{align}
    \mr{fl}(a + b) = (a + b)(1+\tilde{\delta}),\ \; \mr{fl}(a \times b) = (a \times b)(1+\tilde{\delta}),
\end{align}
where $|\tilde{\delta}| \leq u_{\rm fp}$. Let
\begin{align}
    \hat{\tau}_j = \mr{fl}
    \big(d_{t,k,j}^{(I)}\lambda_{j+1}\big),
\end{align}
and let the sequentially computed partial sums be
\begin{align}
    \hat{s}_0=\hat{\tau}_0,\ \;
    \hat{s}_j = \mr{fl}
    \big(\hat{s}_{j-1}+\hat{\tau}_j\big),\ \; 1\leq j\leq J.
\end{align}
Hence, we obtain $\omega_{t,k}^{\rm fp} = \widehat{s}_J$.

Expanding the above recursion, each exact term $\tau_j$ is affected by one rounded multiplication and at most $J$ rounded additions. As a result, the standard floating-point product bound gives
\begin{align}
    \omega_{t,k}^{\rm fp} &= \sum_{j=0}^{J}\tau_j(1+e_j), \\
    |e_j| &\leq \gamma_{J+1} \equiv \frac{(J+1) u_{\rm fp}}{1-(J+1)u_{\rm fp}},
\end{align}
where the standard condition $(J+1)u_{\rm fp}<1$ is assumed. Hence,
\begin{align}
    \Big| \omega_{t,k}^{\rm fp} - \omega_{t,k}^{(I)}\Big| &= \Big| \sum_{j=0}^{J}\tau_j e_j
    \Big| \leq \gamma_{J+1} \sum_{j=0}^{J}|\tau_j|.
\end{align}
Dividing both sides by $|\omega_{t,k}^{(I)}|$ and using the definition of $\mathcal{C}_{t,k}$ in \eqref{Eqn:can_condition}, we obtain
\begin{align*}
    \frac{\big|{\omega}_{t,k}^{\mr{fp}}-\omega_{t,k}^{(I)}\big|}{\big|\omega_{t,k}^{(I)}\big|}\leq \frac{(J_{t,k}+1) u_{\rm fp}}{1-(J_{t,k}+1)u_{\rm fp}} \mathcal{C}_{t,k}.
\end{align*}
This completes the proof.

\bibliographystyle{IEEEtran}
\bibliography{reference}

\end{document}